\documentclass[12pt,oneside,english]{amsart}
\usepackage[T1]{fontenc}
\usepackage[utf8]{inputenc}
\usepackage{geometry}
\usepackage{babel}
\usepackage{textcomp}
\usepackage{bm}
\usepackage{amsthm}
\usepackage{amssymb}
\usepackage{setspace}
\usepackage[authoryear]{natbib}
\usepackage[unicode=true,pdfusetitle,
 bookmarks=true,bookmarksnumbered=false,bookmarksopen=false,
 breaklinks=false,pdfborder={0 0 1},backref=section,colorlinks=true]
 {hyperref}
\hypersetup{
 citecolor=cite-blue}

\makeatletter
\numberwithin{equation}{section}
\numberwithin{figure}{section}
  \theoremstyle{definition}
  \newtheorem{defn}{\protect\definitionname}
  \theoremstyle{remark}
  \newtheorem{rem}{\protect\remarkname}
 \theoremstyle{definition}
  \newtheorem{example}{\protect\examplename}
  \theoremstyle{plain}
  \newtheorem{prop}{\protect\propositionname}
    \theoremstyle{plain}
  \newtheorem{lemma}{\protect\lemmaname}
  \theoremstyle{plain}
  \newtheorem{cor}{\protect\corollaryname}
\theoremstyle{plain}
  \newtheorem{claim}{\protect\claimname}
\newenvironment{lyxlist}[1]
{\begin{list}{}
{\settowidth{\labelwidth}{#1}
 \setlength{\leftmargin}{\labelwidth}
 \addtolength{\leftmargin}{\labelsep}
 }}
{\end{list}}
\theoremstyle{plain}
\newtheorem{thm}{\protect\theoremname}

\usepackage[foot]{amsaddr}

\makeatother

  \providecommand{\definitionname}{Definition}
  \providecommand{\examplename}{Example}
  \providecommand{\lemmaname}{Lemma}
  \providecommand{\propositionname}{Proposition}
  \providecommand{\remarkname}{Remark}
  \providecommand{\corollaryname}{Corollary}
  \providecommand{\claimname}{Claim}
\providecommand{\theoremname}{Theorem}

\begin{document}
\definecolor{cite-blue}{RGB}{0,0,204}
\title{Hiring from a pool of workers}

\date{December 2020}

\author{Azar Abizada}

\author{In\'{a}cio B\'{o}}
\begin{abstract}
In many countries and institutions around the world, the hiring of workers is made through open competitions. In them, candidates take tests and are ranked based on scores in exams and other predetermined criteria. Those who satisfy some eligibility criteria are made available for hiring from a ``pool of workers.'' In each of an ex-ante unknown number of rounds, vacancies are announced, and workers are then hired from that pool. When the scores are the only criterion for selection, the procedure satisfies desired fairness and independence properties. We show that when affirmative action policies are introduced, the established methods of reserves and procedures used in Brazil, France, and Australia, fail to satisfy those properties. We then present a new rule, which we show to be the unique rule that extends static notions of fairness to problems with multiple rounds while satisfying aggregation independence, a consistency requirement. Finally, we show that if multiple institutions hire workers from a single pool, even minor consistency requirements are incompatible with variations in the institutions' rules.
\end{abstract}

\thanks{The authors thank Oguz Afacan, Samson Alva, Julien Combe, Janaína Gonçalves, Rosalia
Greco, Rustamdjan Hakimov, Philipp Heller, Sinan Karadayi, Morimitso
Kurino, Alexander Nesterov, Thibaud Pierrot, Bertan Turhan, Chiu
Yu Ko, and two anonymous referees for helpful comments. Research assistance by Florian Wiek is
much appreciated. Financial support from the Leibniz Association through the SAW project 
MADEP (Bó) is gratefully acknowledged.}

\address{Abizada: School of Business, ADA University, 11 Ahmadbay Aghaoglu St., Baku
AZ1008, Azerbaijan. Email:\href{mailto:aabizada@ada.edu.az}{aabizada@ada.edu.az}.}

\address{Bó (Corresponding author): University of York, Department of Economics and Related Studies, York, United Kingdom; website: \protect\url{http://www.inaciobo.com}; e-mail:
inacio.lanaribo@york.ac.uk.}
\maketitle
\newpage

\section{Introduction\label{sec:Introduction}}

While most companies are free to use almost any criteria to decide which workers to hire and when, that is not the case in many governments
and institutions around the world. To reduce the agency problems
of government institutions and increase the transparency of the hiring
process, those institutions have to follow clear and strict criteria
for selecting workers. In particular, when the number of workers hired
is large (such as police officers, tax agents, etc.), the selection
procedure may consist of several steps, such as written exams, physical
and psychological tests, interviews, and so on, which may also be time
consuming. Due to the high costs of executing such selection procedures,
these hirings often occur in two phases: the evaluation phase,
in which workers apply for the job and take part in the above-mentioned
tests and exams, and the second phase, in which the institutions select,
over time and on a need basis, workers from the ``pool'' of workers
who took part in the first phase. After a certain period,
the pool of workers is renewed, with new ones coming through a new
evaluation phase. As described by the Public Service Commission of
the New South Wales government:
\begin{quotation}
``A talent pool is a group of suitable candidates (whether or not
existing Public Service employees) who have been assessed against
capabilities at certain levels. (...) Using a talent pool enables
you to source a candidate without advertising every time a vacancy
occurs. You can either directly appoint from the pool without further
assessment (for example, to fill a shorter-term vacancy), or conduct
a capability-based behavioral interview with one or more candidates
from the pool to ensure a fit with organizational, team and role requirements
(or additional assessment for agency, role specific or specialist
requirements \textendash{} this is recommended for longer term or
ongoing vacancy). This considerably reduces the time and costs associated
with advertising.''\footnote{Source: Public Service Commission of the New South Wales (http://www.psc.nsw.gov.au/employmentportal/recruitment/recruitment/guide/planning/talent-pools)}
\end{quotation}
The main characteristics of these procedures, which will be essential
to our analysis, are that (i) the selection of workers to hire at
any time, follows a well-defined rule, which is a systematic way of
selecting workers to fill a specified number of positions, (ii) workers
are hired in\emph{ rounds}, on a need-basis, and must be selected
from the pool of workers who took part in the evaluation phase, and
(iii) the institutions do not necessarily know ex-ante the number
of workers they will hire during the pool's validity. Therefore, in general, not all workers in the pool will be hired.
This aspect is emphasized in the description of the selection process
used for all personnel hiring in the European Union institutions:
\begin{quotation}
``The list is then sent to the EU institutions, which are responsible
for recruiting successful candidates from the list. \textbf{Being
included on a reserve list does not mean you have any right or guarantee
of recruitment.}'' {[}emphasis from the original article{]}\cite{euPersonnel}
\end{quotation}
Notice that the quote above refers to a ``reserve list''. A reserve list is a group of candidates who are not hired initially but may (or may not) be hired later. Procedures that mention reserve lists are equivalent to those with a pool of workers: the pool consists of the first set of workers hired together with the reserve list.

A vast number of hirings occur around the world following this type of procedure. Most developed countries use them when hiring public sector workers to the best of our knowledge. Below we provide three examples, which are informative about the number of jobs involved.

All hirings for the U.K. Civil Service occur using open competitions which result in a ``order of merit list'', with a reserve list valid for 12 months.\footnote{Source: Civil Service Commission. 2018. “Recruitment Principles.”} In 2019, there were 445,480 civil servants in the U.K., with 44,570 of them being hired in 2018.\footnote{Source: Civil Service Statistics 2019, Cabinet Office National Statistics} 

In Brazil, the federal constitution mandates that the hiring of public sector workers in all government levels (Federal, State, and Municipal), and state-owned companies, are made through open competition. It moreover states that their results are ``valid for two years'', and that the workers with a non-expired competition result have priority over those with later results.\footnote{Source: Brazilian Federal Constitution (1988), Article 37. In practice, this article implies that after a public competition, workers who ``pass'' the competition are put on hold and might be hired for two years. Those who are not hired have to reapply to a new open competition to be considered.} In 2017, there were more than 11.37 million public sector workers in Brazil.\footnote{Source: Atlas do Estado Brasileiro, IPEA}

In France, most public sector workers' hiring is also made through annual open competitions (\textit{concours}). These result in an order of merit list and must also include a ``complementary list'', with a number of candidates that is at most 200\% of the number of positions hired in the first round.\footnote{Décret n°2003-532 du 18 juin 2003 relatif à l'établissement et à l'utilisation des listes complémentaires d'admission aux concours d'accès aux corps de la fonction publique de l'Etat} In 2018, there were 5.48 million public sector workers in France. In 2016, 40,209 workers at the federal level were hired using these procedures.\footnote{Source: Ministère de l’action et des comptes publics}

Very often, the rules used for hiring workers involve scores in the
selection process. This is not uncommon: the criteria that
are used mostly consist of a weighted average of performance points in multiple
dimensions, such as written exam results, education level, etc.\footnote{Real-life examples of selection rules based on the ranking of workers
are the selection of policemen in Berlin and public sector workers in
Brazil and France.} When the workers' scores constitutes the sole element for determining
which workers to hire, a very natural rule, namely \emph{sequential
priority}, is commonly used: if $q$ workers are to be hired, hire
the $q$ workers with the highest scores among those who remain in the pool.
This rule is simple but has many desirable characteristics. First,
it is fair in the sense that every worker who is not (yet) hired has
a lower score than those who were hired. This adds a vital element
of transparency to the process: if the worker can see, as is often
the case, the scores of those who were hired (or at least the lowest
score among those who were hired), then she has a clear understanding
of why she was not hired. Secondly, it responds to the agency problem:
an institution cannot arbitrarily select low-scoring workers before
selecting all those who have a score higher than that worker. Finally,
the selected workers' quality and identity do not depend on
the number of rounds and vacancies in each round. That is, selecting
20 workers in four rounds with five workers in each results in
selecting the same workers as if 20 workers were selected at once.
We denote this last property by \emph{aggregation independence}. One implication
of this requirement is that the set of selected workers is independent
of the number of rounds and vacancies in each round: selecting 10
workers in two rounds of five workers in each results in the same
selection as selecting two workers in each of five rounds.

While sequential priority satisfies those desirable properties, the criteria used for hiring workers often combine scores with other compositional objectives, in the form of a desired proportion of workers belonging to some subset of the population, such as ethnic minorities, people with disabilities, or women. In section \ref{sec:CompositionalObjectives} we formalize these objectives, noting that these cannot be achieved by the sequential priority rule even if scores are designed to incorporate them, and show that minority reserves, which is arguably the best method for implementing these objectives in static problems, does not satisfy desirable properties when used multiple times over a single pool of workers.

In section \ref{sec:SequentialAdjustedMinorReserves} we present our main contribution, which is a new rule for hiring workers, that is the unique rule that satisfies natural concepts of fairness for this family of problems. We also show that it is essentially the only rule that extends static notions of fairness with compositional objectives to a problem with sequential hirings while being aggregation independent.

In section \ref{sec:HiringRulesAroundTheWorld}, we evaluate rules used in real-life applications in different parts of the world, combining scores with compositional objectives. These include ``quotas'' for individuals with physical or mental disabilities in public sector jobs in France, for black workers in public sector workers in Brazil, and the gender-balanced hiring of firefighters in the Australian province of New South Wales. We show that these fail most of the time to satisfy natural concepts of fairness and aggregation independence. 

Finally, in section \ref{sec:Multiple-institutions} we consider the cases where there are multiple institutions (or locations, departments, etc.) hiring from a single pool of workers. While this scenario is widespread, our main result shows that a mild requirement, saying that the order in which firms hire workers should not change whether some of the institutions hire a worker, essentially leaves us with a single rule, which says that all institutions must hire workers following a single common priority over them. 

Other than the sections described above, the rest of the paper is organized as follows. In section \ref{sec:Hiring-by-rules} we introduce the basic model of hiring by rules and justify the desirability of aggregation independence. In section \ref{sec:Ranking-based-rules} we restrict our focus to rules that are based on scores associated with each worker, in section \ref{sec:SingleRoundHiring} we evaluate the properties of the rules evaluated when they are used for a single round of hiring, and in section \ref{sec:Conclusion} we conclude. Proofs and formal descriptions of the rules absent from the main text are found in the appendix.

\subsection{Related literature}

The structure and functioning of the hiring process for public sector workers have many elements that makes it a clear target for market design: salaries and terms of employment are often not negotiable, the criteria for deciding who should be hired are exogenously given (or ``designed'') and there is a clear concern with issues of fairness and transparency. This paper is, to the extent of our knowledge, the first to evaluate from a theoretical perspective this type of hiring that occurs in the public sector, in which workers are sequentially hired following a predetermined criterion.

There are a few branches of the literature, however, that are related to our analysis. First, the description and analysis of methods for hiring public sector workers around the world and the incentives involved. \cite{sundell_are_2014} evaluates to what extent the use of examinations constitutes a meritocratic method for recruiting in the public sector. The author observes that exams may not be the most adequate way to identify fitness for each function, but that the patronage risk involved when using more subjective criteria such as interviews and CV screening often overcomes those losses. In fact, in an empirical analysis in different ministries in the Brazilian federal government, \cite{bugarin_incentivos_2016} found a positive relationship between corruption cases and the proportion of employees hired by using subjective criteria. 

The property of aggregation independence, which we propose is important for this problem, is related to consistency \citep{thomson1990consistency,tadenuma1991no,thomson1994consistent} notions. Loosely speaking, an allocation rule is consistent if whenever agents leave the problem with their own allocations, the residual problem's solution makes the same allocation among the remaining agents. On the other hand, aggregation independence says that the order (or timing) in which the allocation of a given number of jobs occurs does not change the identity of those who will get the jobs. Different notions of consistency have been used in other matching and allocation problems based on priorities as well \citep{ergin2002efficient,klaus_local_2013}.

Finally, a big part of our analysis concerns what we denote compositional objectives: objectives regarding characteristics that some portions of the workers hired should have, such as a minimum proportion of workers with disabilities, ethnic minorities, or certain genders. \cite{Sonmez2019-hr} evaluated the constitutionally mandated affirmative action policy used in the hiring of public sector workers in India. Similar to the cases we study, workers are also selected based on open competitions that result in an order of merit of the candidates. They also identify some shortcomings that result from how the rules are used to implement affirmative action objectives and propose solutions for them. However, they do not consider the cases in which hirings occur in multiple rounds,\footnote{The Indian civil service, like the Brazilian one, uses reserve lists that are valid for two years. Source: Indian Union Public Service Commission} which, as we show in this paper, might have significant consequences.

Most of the positive and normative literature on the market design consequences of affirmative action policies focus on college admissions and school choice. \cite{Kojima2012-uv} and \cite{hafalir2013effective} evaluate the use of maximum quotas (which limit the number of non-minority students who can be admitted in a school) with minority reserves in the context of a centralized school choice procedure. They show that majority quotas may paradoxically hurt minority students, while minority reserves improve upon this problem. However, these welfare results have no parallel in our analysis, in which workers are either hired or not. Several other papers also evaluate affirmative action procedures currently used to select students into schools or universities, identifying shortcomings and proposing alternative procedures. \cite{aygun2018college} show how the affirmative action procedure used in university admissions in Brazil results in fairness and incentive problems. \cite{Dur2020-ph} studied the allocation of students to Chicago's elite public high schools and compared various reservation policies.\footnote{Other papers, such as \cite{abdulkadirouglu2003school,echenique2015control,bo2016fair,abdulkadirouglu2005college} evaluated affirmative action policies in school and college matching.} All of them, however, treat the problem from a static perspective: either only one choice is made, or a complete allocation is produced once and for all. 

\section{Hiring by rules and aggregation independence\label{sec:Hiring-by-rules}}

A\textbf{ rule} determines which workers an institution should hire,
given a number of workers to hire, a pool of workers, and, potentially,
the workers that the institution hired before. Each time an institution
attempts to hire workers from the pool is denoted as a\textbf{ round}.

Let $A$ be the set of workers hired in previous rounds, and $W$
be the set of workers available. For each $\left(W,A,q\right)$,
a rule $\varphi$ determines which $q$ workers from $W$ should be
hired. That is, for each $\left(W,A,q\right)$,   $\varphi\left(W,A,q\right)\subseteq W$ and $\left|\varphi\left(W,A,q\right)\right|=\min\left\{ q,\left|W\right|\right\} $. For simplicity of notation, we will sometimes use the following shorthand:

\[
\varphi\left(W,A,\langle q_{1},\ldots,q_{t}\rangle \right)\equiv\varphi\left(W^{0},A^{0},q_{1}\right)\cup\varphi\left(W^{1},A^{1},q_{2}\right)\cup\cdots\cup\varphi\left(W^{t-1},A^{t-1},q_{t}\right)
\]

Where $A^{0}=\emptyset$, $W^{0}=W$ and for $i>0$, $A^{i}=A^{i-1}\cup\varphi\left(W^{i-1},A^{i-1},q_{i}\right)$
and $W^{i}=W\backslash A^{i}$. Furthermore, for simplicity, we will use the
shorter notation $\varphi\left(W,q\right)$ when $A=\emptyset$. Unless stated explicitly, none of our results rely on situations in which there are not enough workers, either in general or with some characteristics, to be hired. That is, in all of our results, we will assume that the number of workers in $W$ is at least as large as $\sum q_i$ , and the same holds for the cases that we will evaluate in which some workers belong to minority groups.

One crucial property of the process of hiring by rules is that the
sequence of hires $\lambda=\langle q_{1},\ldots,q_{t}\rangle$ is ex-ante unknown. That
is, every round may or may not be the last one. The total number of
workers who will be hired is also unknown. Therefore, the properties
that we will deem as desirable should hold at any point in time. In
this context, a critical property that a rule should satisfy is\textbf{
aggregation independence}. A rule is aggregation independent if the
total set of workers hired after a certain number of rounds does not
depend on how they are distributed among rounds.
\begin{defn}
A rule $\varphi$ is\textbf{ aggregation independent} if for any $q\geq q_{1}\geq0$
and sets of workers $W$ and $A$, $\varphi\left(W,A,q\right)=\varphi\left(W,A,\langle q_{1},q-q_{1}\rangle \right)$.
\end{defn}
Therefore, when the rule being used is aggregation independent, an
institution that hires $q_{1}$ workers in the first round and $q_{2}$
in the second will select the same workers that it would by hiring
$q_{1}+q_{2}$ in a single round. One can easily check that
if a rule is aggregation independent, this extends to any combination
of rounds: if $\sum_{i}q_{i}=\sum_{j}q_{j}'$, a sequence of
hires $q_{1},\ldots,q_{n}$ will select the same workers as
$q_{1}',\ldots,q_{m}'$.

We now provide three reasons to justify aggregation independence as
a strongly desired property for rules for hiring by rules: transparency,
non-manipulability, and robustness.

\textbf{Transparency}

One of the main reasons driving governments and institutions to
use hiring by rules is that, for those who are not hired, the reason that happens is made clear and straightforward. For example, take the
rule that consists of always hiring the workers with the highest exam scores. By knowing the rule and observing the hired workers (and their scores), any worker who was not hired knows that
there was no obscure reason why she was not yet hired: it is merely
because her score was lower than all those who were hired.

Suppose, however, that the rule that is used is not aggregation independent.
Then, a worker who was not hired, by just looking at the set of workers
who were hired, may not be able to easily understand why she was not
hired, even understanding the rule that was used, because it would also be necessary for her to know the precise sequence of the number of
workers that were hired in each round.

\textbf{Non-manipulability}

While many times the rules which govern the hiring process are chosen
in a way that reduces the ability of managers to make arbitrary choices
of whom to hire, they may have freedom in choosing the sequence of
hires. For example, instead of hiring four workers in one month, she
may choose to hire two workers first and then two additional workers.

If the rule is aggregation independent, different choices
of sequences of workers hired will not lead to different sets of workers
hired. However, if the rule is not aggregation independent, that may
not be the case, and a manager may choose a specific sequence of hiring
decisions, which will allow a particular worker to be hired, whereas
she would not be, absent the specific sequence chosen. An aggregation
independent rule, by definition, is not manipulable by the choice
of the sequence of hires.

\textbf{Robustness}

The third reason why aggregation independence is a desirable property
is that the degree to which the set of workers hired satisfies the
objectives represented in the rule is robust to uncertainty or bad
planning on the part of the manager in terms of the number of workers
that is needed. In other words, assuming that the criterion for choosing
workers which is set by the rule represents the desirability of the
workers it chooses (for example, it chooses the most qualified set
of workers subject to some constraint), an aggregation independent
rule will always choose the best set of workers, whether the manager
makes hiring decisions all at once or continually re-evaluates the
number of workers to be hired. Aggregation independent rules do not
have that problem: managers may hire workers based on demand,
and that will not result in a less desirable set of workers hired.

In Section \ref{subsec:SequentialMinorityReserves} we show specific
examples of how aggregation independence relates to non-manipulability
and robustness.

\section{Score-based rules\label{sec:Ranking-based-rules}}

A common way workers are selected when hiring by rules is
through a scoring of all workers. Using criteria such as written exams,
evaluation of diplomas, certificates, and experience, workers receive
a score (or a number of points). These scores become the deciding factor
of who to hire: when hiring $q$ workers, hire the $q$ workers
with the highest scores from the pool. For a set of workers $W$,
let $\bm{s_{W}}=(s_{w})_{w\in W}$ be the score profile of workers
in $W$, where for all $w\neq w'$, $s_w\neq s_{w'}$.\footnote{Our assumption that no two workers have the same score is based on how the procedures that we consider work in real life. Even when discretized scores are used for evaluating the candidates, the process always results in a strict ordering of these workers. This can be seen in how the legislation refers to an ``order of merit''', or in the details of the hiring posts, which often describe multiple (deterministic) methods for breaking ties.} Denote by $top_{q}(W)$ the $q$ workers with the highest scores in $W$.\footnote{Although $W$ is a set, for simplicity of notation we will consider
$s_{W}$ following the order in which the elements of $W$ are written.
For example, if we denote $W=\left\{ w_{2},w_{1},w_{3}\right\} $
, $s_{W}=\left(10,20,30\right)$ implies that worker $w_{2}$ has
a score of $10$. } A natural property for a score-based rule is for it to be\textbf{
fair}. That is, after any number of rounds, if a worker $w$ was hired
and $w'$ was not, then $s_{w}>s_{w'}$.
\begin{defn}
A rule $\varphi$ is\textbf{ fair} if for any $W$, $A$ and $\lambda=\langle q_{1},\ldots,q_{t}\rangle $,
$w\in\varphi\left(W,A,\lambda \right)$
and $w'\not\in\varphi\left(W,A,\lambda \right)$
implies that $s_{w}>s_{w'}$.
\end{defn}
A natural rule for these kinds of problems is what we denote by\textbf{
sequential priority}. When hiring $q$ workers, it consists of selecting
the $q$ workers with the highest score from the pool of workers.
If the pool contains less than $q$ workers, then hire all of them.
The following remark comes immediately from the definition of the
rule.
\begin{rem}
The sequential priority rule is aggregation independent and fair.
\end{rem}
When the selection of workers is based on scores, which is a very
typical setup, the sequential priority rule gives us all we need: it is
fair and aggregation independent.

\section{Compositional objectives\label{sec:CompositionalObjectives}}

It is common for hiring processes based on rules to combine the
use of scores with compositional objectives, such as affirmative action.
Typically, the objective is to reserve some of the jobs for workers
with a certain characteristic, sometimes those belonging to an ethnic minority
or those who possess some type of disability. Denote by $M$ the set
of workers who belong to the minority group (that is, $M\subseteq W$)
and $\omega\left(W\right)$ be the number of minority workers in $W$. The
affirmative action policy also has a\textbf{ minority ratio} $m$,
where $0\leq m\leq1$, which represents the proportion of hires that
should be based on affirmative action.

As argued in Section \ref{sec:Hiring-by-rules}, the desirable properties
associated with affirmative action should also hold after any number
of rounds. Our first requirement is that, when possible, the proportion
of selected minorities should be at least $m$ after each round.
\begin{defn}
\label{def:RespectsMinorityRights}
A rule $\varphi$\textbf{ respects minority rights} if, for any $W$
and sequence of hires $\lambda=\langle q_{1},\ldots,q_{t}\rangle $, $(i)$
when $\left|M\right|\geq m\times\sum_{i=1}^{t}q_{i}$ we have $\omega\left(\varphi\left(W,\lambda \right)\right)/\left|\varphi\left(W,\lambda \right)\right|\geq m$,
or $(ii)$ when $|M|<m\times\sum_{i=1}^{t}q_{i}$ we have $M\subset\varphi\left(W,\lambda \right)$.
\end{defn}
\begin{rem}
The sequential priority rule does not respect minority rights. In
general,\textit{\emph{ fairness}} is incompatible with\textit{\emph{
respecting minority rights}}\emph{.}\footnote{Assume that there are three workers: one minority (call him K) and
two non-minority (L and V), where the scores are as follows $s_{L}>s_{V}>s_{K}$
. If the rule needs to select two workers and $m$ is 0.3, then in
order to\textit{ respect minority rights}, the rule should select
K and L, which is not\textit{ fair}, as\textit{ fairness} requires
L and V to be selected.}
\end{rem}
Therefore, we define a weaker notion of fairness, which takes into
account the minority restriction. A rule is minority fair if, conditional
on respecting minority rights, the hiring decision is based on
scores.
\begin{defn}
\label{def:MinorityFair}
A rule $\varphi$ is\textbf{ minority fair} if, for any $W$, $M\subseteq W$
and $\lambda=\langle q_{1},\ldots,q_{t}\rangle $, where $H=\varphi\left(W,\lambda \right)$:

$(i)$ for each $w,w'\in W\setminus M$ or $w,w'\in M$, if $w\in H$
and $w'\notin H$, then $s_{w}>s_{w'}$,

$(ii)$ for each $w\in W\setminus M$ and $w'\in M$, if $s_{w}<s_{w'}$
and $w\in H$, then $w'\in H$,

$(iii)$ if there is $w\in W\setminus M$ and $w'\in M$ with $s_{w}>s_{w'}$,
$w\notin H$ and $w'\in H$, then $\omega\left(H\right)/\left|H\right|\leq m$.
\end{defn}
In words, a rule is minority fair if it (i) chooses between workers
from the same group (minorities or non-minorities) based on their
scores, (ii) does not hire low-scoring non-minorities while higher-scoring
minorities are available, and (iii) only hires low-scoring minorities
over higher-scoring non-minorities when that is necessary to bring
the ratio of minorities closer to $m$ from below. 

One natural question one might have is whether a carefully designed ``standardized'' scoring system could turn the sequential priority rule into one that respects minority rights and/or is minority fair. In other words, is there a transformation $s_W'$ of the scores $s_W$, for which the sequential priority rule under $s_W'$ respects minority rights and is minority fair under $s_W$? The answer for the first part is yes, and for the second, no. To see this, let $\overline{s}$ be the highest score any worker might have, and $s_W'$ be a scoring system that replicates $s_W$ for non-minorities, and that adds $\overline{s}$ to the scores of minorities. That is, $s_W'$ makes all minorities have higher scores than non-minorities but keep all relative scores the same otherwise. It is clear that the sequential priority rule under this scoring system respects minority rights but might not be minority fair since it will always hire minority candidates whenever some are still available, regardless of the value of $m$.

To see that no transformation of a scoring system can respect minority rights and be minority fair, consider the following problem: $W=\left\{w_1,w_2,w_3,w_4\right\}$, where $M=\left\{w_2,w_4\right\}$, $s_W=\left(100,90,80,70\right)$ and $m=0.5$. Let $s_W'$ be the scoring of the workers in $W$ derived from this transformation. Since the sequential priority rule under $s_W'$ respects minority rights and is minority fair, it is easy to check that it must satisfy $s_{w_2}>s_{w_1}>s_{w_4}>s_{w_3}$. Imagine, however, that the set of workers is $W=\left\{w_1,w_3,w_4\right\}$. In this case, the sequential priority under $s_W'$ would imply that if $q_1=1$, worker $w_1$ would be hired, which is a violation of minority rights since it would require $w_4$ to be hired instead.\footnote{An alternative question that one might ask is whether there is a scoring function, defined for each given set of workers and original scores, which respects minority rights and is minority fair. The answer for this one will be yes, since one can produce a ranking of workers from any aggregation independent rule by choosing one worker at a time, and as we will show later in this paper, such a rule exists. In terms of interpretation, we believe that a scoring rule that is endogenous to the set of workers being considered is more of a technical property that results from aggregation independence than a criterion that can be described as a scoring method for hiring before knowing who will apply for the jobs.}


\subsection{Sequential use of minority reserves\label{subsec:SequentialMinorityReserves}}

Perhaps the most natural candidate for a hiring rule for this family of problems is the use of \textit{reserves}. With this method, in each period
in which there are vacancies to be filled, the institution uses a
\textit{choice procedure generated by reserves} \citep{hafalir2013effective,echenique2015control},
with a proportion $m$ of vacancies reserved for minority workers.

Given a set of workers $W$, of minorities $M\subseteq W$, a number
of reserved positions $q^{m}$ and of hires $q$, a choice generated
by reserves consists of hiring the top $\min\left\{ q^{m},\left|M\right|\right\} $
workers from $M$ and then filling the remaining $q-\min\left\{ q^{m},\left|M\right|\right\} $
positions with the top workers in $W$ still available. In a static
setting, this procedure is shown to have desirable fairness and efficiency
properties while satisfying the compositional objectives \citep{hafalir2013effective}.
We denote the sequential use of minority reserves rule by $\varphi^{SM}$.

In our setting, therefore, the sequential use of minority reserves
rule consists of, in a round $r$, hiring $q_{r}$ workers, reserving
$m\times q_{r}$ of them for minorities.
\begin{prop}
\label{prop:SeqMinorityReservesRespectsMinorityRightsNotConsistentNotFair}The
sequential use of minority reserves respects minority rights. However,
it is neither minority fair nor aggregation independent.
\end{prop}
The next example shows the problems associated with this rule.
\begin{example}
\label{example:sequentialMinorityReserves}Let $W=\left\{ w_{1},w_{2},w_{3},w_{4},w_{5}\right\} $,
$M=\left\{ w_{1},w_{2},w_{5}\right\} $, and $s_{W}=\left(100,90,80,50,20\right)$.
Let $r=2$, $q_{1}=q_{2}=2$ and $m=0.5$. In the first round, the
top worker from $M$ and the top from $W\backslash\left\{ w_{1}\right\} $
are hired, that is, $\left\{ w_{1},w_{2}\right\} $. In the second
round, the pools of remaining workers are $W^{2}=\left\{ w_{3},w_{4},w_{5}\right\} $
and $M^{2}=\left\{ w_{5}\right\} $. The top worker from $M^{2}$
, that is, $\left\{ w_{5}\right\} $, and the top from $W^{2}\backslash\left\{ w_{5}\right\} $
are hired. Therefore, $\left\{ w_{3},w_{5}\right\} $ are hired in
the second round and $\varphi^{SM}\left(W,\langle q_{1},q_{2}\rangle \right)=\left\{ w_{1},w_{2},w_{3},w_{5}\right\} $.

Now consider the case where $q=q_{1}+q_{2}=4$. Then in the first
and unique round, the two top workers from $M$, $\left\{ w_{1},w_{2}\right\} $,
and the top two workers from $W\backslash\left\{ w_{1},w_{2}\right\} $
are hired, that is, $\left\{ w_{3},w_{4}\right\} $. Therefore, $\varphi^{SM}\left(W,q_{1}+q_{2}\right)=\left\{ w_{1},w_{2},w_{3},w_{4}\right\} $.
Therefore, the $\varphi^{SM}$ rule is\emph{ not aggregation independent}.
Moreover, note that $w_{4}\in\varphi^{SM}\left(W,\langle q_{1},q_{2}\rangle \right)$,
$w_{5}\in\varphi^{SM}\left(W,\langle q_{1},q_{2}\rangle \right)$
and $s_{w_{4}}=50>20=s_{w_{5}}$ while $\omega\left(\varphi^{SM}\left(W,\langle q_{1},q_{2}\rangle \right)\right)/\left|\varphi^{SM}\left(W,\langle q_{1},q_{2}\rangle \right)\right|=0.75>0.5=m$,
implying that the $\varphi^{SM}$ rule is\emph{ not minority fair}.
\end{example}
The sequential use of minority reserves rule is a good rule for providing
examples of problems associated with rules that are not aggregation
independent. First, consider the issue of \textbf{manipulability}.
Take the example \ref{example:sequentialMinorityReserves} above and suppose
that the manager prefers to hire the worker $w_{5}$. If she hires the four workers that she
needs all at once, $w_{5}$ would not be hired. If, instead, she
chooses to hire first two workers, and then later two more workers,
$w_{5}$ will be hired. That is, by choosing a sequence of hires strategically, the manager can hire the person she wanted. 

Next, we show that the lack of aggregation independence may lead to
the hiring of a group of workers who are not in line with some common
objectives of desirability, (the issue of \textbf{robustness}, as
described in section \ref{sec:Hiring-by-rules}). To see how this
can be a problem, consider the example below:
\begin{example}
Let $W=\{w_{1},w_{2},w_{3},w_{4},w_{5},w_{6},w_{7},w_{8}\}$, $M=\{w_{5},w_{6},w_{7},w_{8}\}$,
$s_{W}=(100,90,80,70,60,50,40,30)$ and $m=0.5$.

If workers are hired in four rounds, where $q_{1}=q_{2}=q_{3}=q_{4}=1$,
the set of workers hired will be $\{w_{5},w_{6},w_{7},w_{8}\}$. If,
on the other hand, workers are hired all at once, with $q_{1}=4$,
the set of workers hired will be $\{w_{1},w_{2},w_{5},w_{6}\}$
\end{example}
Assuming that the scores are a good representation of the degree of desirability
of a worker for a task, the example above shows that a lack of planning
could lead to hiring a set of workers that are substantially less qualified.

\section{Sequential adjusted minority reserves\textbf{\label{sec:SequentialAdjustedMinorReserves}}}

We now present a new rule, \textit{sequential adjusted minority reserves},
denoted by $\varphi^{SA}$. It consists of the sequential minority
reserves rule in which the number of vacancies reserved for minorities
is adjusted in response to hires made in previous rounds. More specifically,
the rule works as follows:\footnote{For simplicity, the description below assumes that the number of workers
in $M$ and $W$ is large enough so that in every round there is a
sufficient number of them to be hired. A more general description
can be found in the appendix.}
\begin{lyxlist}{00.00.0000}
\item [{\textbf{Round 1}}] Let $m^{1}=m$, $M^{1}=M$ and $W^{1}=W$.
The top $m^{1}\times q_{1}$ workers from $M^{1}$ are hired. Denoted
those workers by $A^{*}$. Additionally, the top $\left(1-m^{1}\right)\times q_{1}$
workers from $W^{1}\backslash A^{*}$ are hired. Let $M^{2}$ be the
workers in $M^{1}$ who were not yet hired, and $W^{2}$ be the workers
in $W^{1}$ who were not yet hired.
\item [{\textbf{Round $r\geq 1$}}] Let $A^{r}=\varphi^{SA}\left(W,\langle q_{1},\ldots,q_{r-1}\rangle \right)$
and $m^{r}=\max\left\{ m-\frac{\omega\left(A^{r}\right)}{\sum_{i=1}^{r}q^{i}},0\right\} $.
The $m^{r}\times q_{r}$ top scoring minority workers in $M^{r-1}$
are hired. Denote those workers by $A^{*}$. Additionally, the top
$\left(1-m^{r}\right)\times q_{r}$ workers from $W^{r}\backslash A^{*}$
are hired. Let $M^{r+1}$ be the workers in $M^{r}$ who were not
yet hired, and $W^{r+1}$ be the workers in $W^{r}$ who were not
yet hired.
\end{lyxlist}
Therefore, the sequential adjusted minority
reserves adapts the set of workers hired according to those
hired in previous rounds. This makes sense: if we do not take into
account, for example, that after the last round, the number of minority
workers greatly exceeded the minimum required, some high-scoring non-minority
workers may not be hired, leading to a violation of minority fairness.
The theorem below shows that this is essentially the only way of achieving these objectives.

\begin{thm}
\label{thm:MSPMREquivMinorityRIghtsFairConsistent} If $\varphi$ is a rule that is minority fair and respects minority rights, then for every set of workers $W$ and sequence of hires $\lambda$, $\varphi(W,\lambda)=\varphi^{SA}(W,\lambda)$.
\end{thm}

The sequential adjusted minority reserves is also the only rule that extends ``static'' notions of fairness to this dynamic setting while being aggregation independent. To see that, we first define the static counterparts of definitions \ref{def:RespectsMinorityRights} and \ref{def:MinorityFair}.

\begin{defn}
\label{def:RespectsStaticMinorityRights}
A rule $\varphi$\textbf{ respects static minority rights} if, for any $W$
and $q> 0$, $(i)$
when $\left|M\right|\geq m\times q$ we have $\omega\left(\varphi\left(W,\langle q \rangle \right)\right)/\left|\varphi\left(W,\langle q \rangle \right)\right|\geq m$,
or $(ii)$ when $|M|<m\times q$ we have $M\subset\varphi\left(W,\langle q \rangle \right)$.
\end{defn}

\begin{defn}
\label{def:StaticMinorityFair}
A rule $\varphi$ is\textbf{ static minority fair} if, for any $W$, $M\subseteq W$
and $q>0$, where $H=\varphi\left(W,\langle q \rangle \right)$:

$(i)$ for each $w,w'\in W\setminus M$ or $w,w'\in M$, if $w\in H$
and $w'\notin H$, then $s_{w}>s_{w'}$,

$(ii)$ for each $w\in W\setminus M$ and $w'\in M$, if $s_{w}<s_{w'}$
and $w\in H$, then $w'\in H$,

$(iii)$ if there is $w\in W\setminus M$ and $w'\in M$ with $s_{w}>s_{w'}$,
$w\notin H$ and $w'\in H$, then $\omega\left(H\right)/\left|H\right|\leq m$.
\end{defn}

In other words, a rule satisfies the static versions of these two notions if they hold when there is only one round of hiring. As a result, a rule that respects minority rights also respects static minority rights, and a rule that is minority fair is also static minority fair. However, the converse does not apply: definitions \ref{def:RespectsStaticMinorityRights} and \ref{def:StaticMinorityFair} restrict only the first set of workers hired in a sequence of hires.

The following result shows that the sequential adjusted minority reserves is essentially\footnote{Since the definition of a rule includes an arbitrary set of previous hires, there can be more rules that can satisfy those properties by defining them differently, for example, for sets of previous hires that could not result from using the rule from the beginning. If we restrict ourselves to the case in which $A=\emptyset$ (the rule is used for every hire ever made), this is a uniqueness result.} the only rule that extends these static notions to multiple hires while being aggregation independent.

\begin{thm}
\label{thm:SAMAggIndep} 
If $\varphi$ is a rule that respects static minority rights, is static minority fair, and aggregation independent, then for every set of workers $W$ and sequence of hires $\lambda$, $\varphi(W,\lambda)=\varphi^{SA}(W,\lambda)$.
\end{thm}

Moreover, notice that the first round of hiring in the sequential adjusted minority reserves is equivalent to the one in the sequential use of minority reserves. Therefore, the sequential adjusted minority reserves is also the aggregation independent extension of the static \textit{minority reserves} rule \citep{hafalir2013effective}.

\begin{cor}
\label{cor:SAMIsMinorityReserves+AI}
If $\varphi$ is a rule for which $\varphi(W,\langle q\rangle)=\varphi^{SM}(W,\langle q\rangle)$ and $\varphi$ is aggregation independent, then for every set of workers $W$ and sequence of hires $\lambda$, $\varphi(W,\lambda)=\varphi^{SA}(W,\lambda)$.
\end{cor}

\section{Hiring rules around the world}
\label{sec:HiringRulesAroundTheWorld}

In the following sections, we present the rules currently being used
in France, Brazil, and Australia, and show that they suffer from different
issues.  

\subsection{Public sector workers in France\label{subsec:FrenchCase}}

By law, every vacancy in the French public sector must be filled through
an\emph{ open competition}. When vacancies are announced, a document
explaining deadlines, job specifications, and the criteria that will
be used to rank the applicants is published. Workers who satisfy some
stated requirements then take written, oral, and/or physical
exams. In some cases, diplomas or other certifications can also be
used for evaluating the workers. At the end of this process, all workers' results in these tests are combined in a predetermined
way, to produce a ranking over all workers. If the number of vacancies
announced was $q$, then the top $q$ workers are hired. An additional
number of workers are put on a ``waiting list.'' These workers may be hired if some of the top $q$ workers reject
the job offer or if additional vacancies need to be filled before
a new open competition is set.

The French law also establishes that at least 6\% of the vacancies
should be filled by people with physical or mental disabilities. Instead
of incorporating the selection of those workers into the hiring procedure
in a unified framework, the institutions instead open, with unclear
regularity, vacancies exclusive for workers who have those disabilities.\footnote{In some cases, different procedures are also used, such as making candidates with disabilities compete for the same vacancies as those without disabilities, but giving them a ``bonus'' in their scores. Similar methods are used in affirmative action policies elsewhere, such as in local universities in Brazil. In this paper, however, we focus on examples involving quotas and reserved vacancies.} 
The hiring of workers over time continues following the same procedure
as the open positions described above. However, nothing prevents workers
with disabilities from applying for open positions. In
fact, the authorities provide some accommodation for these workers
during the selection, such as, for example, allowing for extra time
to write down the exams. These are meant as an attempt to make
up for some disadvantages that those workers have with respect to
those without disabilities, and not to give any advantage.

Let $W^{*}$ be the set of workers who applied for the open competitions, and $M^{*}$ those who applied for the competitions reserved for candidates with disabilities. Workers in $M^{*}$ \textit{must} have a disability, but workers with disabilities might also apply for open competitions. Therefore, $M^{*}\subseteq M$, and in general $M^{*}\cap W^{*}\neq\emptyset$. Since these constitute different competitions, there are scores for the workers in each competition, and more specifically, a worker who applies to both competitions might obtain different scores in each. Therefore, we denote by $s_{w}^{O}$ the score obtained by worker $w$ in the open competition and by $s_{w}^{D}$ the score that worker $w$ obtained in the competition for workers with disabilities. Since $W^{*}$ and $M^{*}$ are usually different sets of workers, some workers might have a score in only one competition, but some workers with disabilities might have in both.

The number of vacancies that are open for workers with disabilities, and when they are opened, is not determined by any law and is mostly done in an ad hoc manner. To evaluate the consequences of the method used in France in a formal way, however, we will consider two alternative policies.\footnote{We have no evidence that any of these policies constitute actual practice by French institutions, but we believe that they represent the two most natural attempts at satisfying the legal requirements under the current rules.} In both cases, we will assume that there are two pools of workers, $W^{*}$ and $M^{*}$, and given a number of positions to be hired $q$, a total of $q$ workers from these pools must be hired. In what follows, we consider an arbitrary sequence of hires $\langle q_1,q_2,\ldots,\rangle$.

\textbf{Policy 1}: This policy consists of first hiring the top $q_{1}$
workers from $W^{*}$ and then adjusting the number of workers in
$M^{*}$ hired in later rounds. For example, say that $q_{1}=100$,
but only four workers among the top $100$ workers in $W^{*}$ (with
respect to $s_{W}^{O}$) hired have disabilities. Then, considering
the objective of hiring at least 6\% workers with disabilities, if
$q_{2}=50$, then open five vacancies exclusive for workers in $M^{*}$
(selected with respect to $s_{W}^{D}$) and $45$ for those in $W^{*}$
(selected with respect to $s_{W}^{O}$). As a result, by the end of
the second round, at least $9$ workers with disabilities, or $m\times(q_{1}+q_{2})$,\footnote{For simplicity, here and in the rest of the main text we will assume
that every expression involving numbers of workers or vacancies are
integers. In the appendix, we relax that restriction and show that
none of the results presented depend on that.} will be hired.

\textbf{Policy 2}: This policy consists of first hiring $m\times q_{1}$
from $M^{*}$ (selected with respect to $s_{W}^{D}$), $(1-m)\times q_{1}$
workers from $W^{*}$ (selected with respect to $s_{W}^{O}$) and
then adjusting the number of workers in $M^{*}$ hired in later rounds.
For example, say that $q_{1}=100$. Then the policy will result in
hiring six workers from $M^{*}$ and $94$ from $W^{*}$ in the first
round. At least $6\%$ of the workers hired would be among those with
a disability, therefore, but potentially more. Suppose that eight workers
with disabilities were hired in the first round and that $q_{2}=50$.
Then in the second round, two vacancies exclusive for workers in
$M^{*}$ would be open, and the remaining $48$ would be open for
all workers in $W^{*}$.

While the two policies above represent what we believe are the best efforts to satisfy the objectives stated in the law under the current existing procedures, Policy 1 differs in that when only one round of hiring is done, the proportion of workers with a disability hired might be below the minimal proportion stated in the law. This fact will more evident in Proposition \ref{prop:FrenchAssignment}.

Whenever necessary, we will refer to the rules defined by policies 1 and
2 by $\varphi^{F_{1}}$ and $\varphi^{F_{2}}$. Since under the French
assignment rule each worker may have one or two scores, what constitutes
minority fairness is less clear in this context. However, the example below
shows that policy 2 may lead to outcomes that clearly violate
the spirit of minority fairness.

\begin{example}
Let $W^{*}=\{w_{3},w_{4},w_{5}\}$ and $M^{*}=\{w_{1},w_{2},w_{3}\}$,
with scores $s_{W}^{O}=(50,40,30)$ and $s_{W}^{D}=(50,40,30)$ and
$m=0.5$. If $q=2$, then $\varphi^{F_{2}}\left(\left\{ W^{*},M^{*}\right\} ,q\right)$
will select $\left\{ w_{1},w_{3}\right\} $. Worker $w_{2}$, however,
has a disability and a better score than $w_{3}$ in the competition
in which both participated.

Notice that if worker $w_{2}$ also applied for the open vacancies
and in that competition obtained a score that is also better than
the one obtained by $w_{3}$, she would have been hired instead of
$w_{3}$. If the relative rankings of the workers in both competitions
are different, more intricate violations of the spirit of minority
fairness can also occur. If we make the (strong) assumptions
that all workers with disabilities apply to both competitions and
that the relative rankings between those workers in both competitions
are the same, we can obtain a clear distinction between both
policies, as shown below.
\end{example}
\begin{prop}
\label{prop:FrenchAssignment} Suppose that $M^{*}\subseteq W^{*}$
and that for every $w,w'\in M^{*}$, $s_{w}^{O}>s_{w'}^{O}\iff s_{w}^{D}>s_{w'}^{D}$.
Policy 1 of the French assignment rule does not respect minority rights
and is not aggregation independent. Policy 2 respects minority rights,
is aggregation independent, and minority fair.
\end{prop}
 It is crucial to notice, however, that the result in proposition \ref{prop:FrenchAssignment}
depends on the \textbf{relative rankings of the workers with disabilities
being the same in both competitions}, but perhaps most importantly,
on\textbf{ workers with disabilities participating in both competitions}.
This is not a minor issue, since these competitions often involve
a significant amount of time and effort.

\subsection{Quotas for black public sector job workers in Brazil\label{subsec:BrazilCase}}

The rules for the hiring of public sector workers in Brazil work
essentially in the same way as in France: vacancies are filled with
open competitions that result in scores associated with the workers,
and workers are hired in each period by following their scores in
descending order. Differently from France, however, there is no quota
for workers with disabilities, but instead, since 2014, there are
quotas for black workers.

The use of racial and income-based quotas has increased significantly
in many areas of the Brazilian public sector and higher education.
At least 50\% of the seats in federal universities, for example, are
reserved for students who are black, low-income, or studied in a public
high-school \citep{aygun2013college}. Many municipalities also employ
quotas for black workers in the jobs that they offer. One of the most
significant recent developments, however, is a law which establishes
that 20\% of the vacancies offered in each job opening should give
priority to black workers.\footnote{Lei N. 12.990, de 9 de junho de 2014.}

Differently from France, the quotas for black workers are explicitly
incorporated into the hiring process. More specifically, the rule
currently used in Brazil (denoted the $\varphi^{B}$ rule) works as
follows. Let $k$ be a number that is higher than any expected number
of hires to be made.
\begin{lyxlist}{00.00.0000}
\item [{\textbf{Initial step}}] Workers are partitioned into two groups:
(i) Top Minority ($TM$) and (ii) Others ($O$). The $TM$ group consists
of the highest scoring $\lceil m\times k \rceil$ workers from $M$, and $O$
be the top $k-\lceil m\times k \rceil$ workers in $W\backslash TM$.\footnote{Notice that the set $W\backslash TM$, in general, contains both minority and non-minority workers. As a result, if there are not enough minority workers, the remaining positions
are filled with the top non-minority workers.} Let $TM^{1}=TM$ and $O^{1}=O$.
\item [{\textbf{Round $r\geq 1$}}] The $\lceil m\times q_{r}\rceil$ top scoring
minority workers from $TM^{r}$, and the top $q-\lceil m\times q\rceil$ workers
from $O^{r}$ are hired. By removing these workers hired, we obtain
$TM^{r+1}$ and $O^{r+1}$.
\end{lyxlist}
For the Brazilian law specifically, $m=0.2$. The example below shows that the Brazilian rule is not minority fair.

\begin{example}\label{example:BrazilNotMinorityFair}
Let $W=\left\{ w_{1},w_{2},w_{3},w_{4}\right\} $, $M=\left\{ w_{1},w_{2}\right\} $,
and $s_{W}=\left(100,90,80,50\right)$. Let $q=2$, $m=0.5$ and $k=4$.
Then $TM^{1}=\left\{ w_{1},w_{2}\right\} $ and $O^{1}=\left\{ w_{3},w_{4}\right\} $.
The Brazilian rule states that, when hiring two workers, the top
worker from $TM^{1}$ and the top from $O^{1}$ should be hired. Therefore,
$\varphi^{B}\left(W,q\right)=\left\{ w_{1},w_{3}\right\} $. Since
$w_{2}\notin\varphi^{B}\left(W,q\right)$ and $s_{w_{2}}>s_{w_{3}}$,
the Brazilian rule is not\textit{ minority fair}.\footnote{One may conjecture that the scenario above is very unexpected, since
the affirmative action law must have been enacted in response to minority
workers not being hired based solely on scores. As shown in \cite{aygun2013college},
however, this conjecture may be misleading. For example, even if the
average score obtained by minority workers is lower, one can have
situations in which the preferences of the higher achieving minority
workers are correlated, leading to the top minority workers in the
entire population applying to a specific job.}
\end{example}
Notice that in this example, worker $w_{2}$, who is part of the minority,
has a higher score than $w_{3}$, who is not a minority. Worker $w_{3}$
is hired, while $w_{2}$ is not. Given that the affirmative action
rules were designed with the intent of increasing the access that
minorities have to these jobs, this type of lack of fairness is especially
undesirable, since if the hiring process was purely merit-based, worker
$w_{2}$ would have been hired.

\cite{aygun2013college} describe the implementation of affirmative
action in the admission to Brazilian public universities. There, as
here, the problems arise from the fact that positions (in that case
seats) and workers are partitioned between those reserved for minorities
and the open positions. Differently from there, however, unfair outcomes
may not be prevented by workers even if they strategically manipulate
their minority status. In the example above, even if $w_{2}$ applied
as a non-minority he would not be hired.
\begin{prop}
\label{prop:BrazilianRuleNotMinorityFair}The Brazilian rule is aggregation
independent and respects minority rights. However, it is not minority
fair.
\end{prop}

\subsection{Gender balance in the hiring of firefighters in New South Wales\label{subsec:NSWGenderEqualityRule}}

The hiring of firefighters in the Australian province of New South
Wales attempts to achieve a gender-balanced workforce by following
a simple rule:
\begin{quote}
``Candidates who have successfully progressed through the recruitment
stages may then be offered a place in the Firefighter Recruitment
training program. Written offers of employment will be made to an
equal number of the most meritorious male and female candidates based
on performance at interview and the other components of the recruitment
process combines.''\footnote{Source: Fire \& Rescue NSW (https://www.fire.nsw.gov.au/page.php?id=9126)}
\end{quote}
We denote this rule by\emph{ NSW rule}, or $\varphi^{NSW}$. Although not stated explicitly in the institution's website, we will assume that if there are not enough individuals of some gender, the remaining hirings will be made among those candidates available, based on their scores. Moreover, to avoid results that rely simply on whether $q$ is odd or even, we assume it is always even. The example
below shows the problems involved in that rule:
\begin{example}\label{example:NSWNotfairNotMinorityFair} 
Let $W=\left\{ w_{1},w_{2},w_{3},w_{4}\right\} $, and $W^{F}=\left\{ w_{1},w_{2}\right\} $
and $W^{M}=\left\{ w_{3},w_{4}\right\} $ be the set of female and
male workers, respectively. Suppose that the scores are $s_{W}=\left(100,90,80,50\right)$.
Let $q=2$. Then $\varphi^{NSW}\left(W,q\right)=\left\{ w_{1},w_{3}\right\} $.
Since $w_{2}$ is not hired but $s_{w_{2}}>s_{w_{3}}$, the NSW rule
is not fair. Moreover, it is easy to see that if either gender is considered
a minority, the rule is also not minority fair.
\end{example}
The result below summarizes the properties of the NSW rule.
\begin{prop}
The NSW rule is aggregation independent but not fair. If one of the
genders is deemed as a minority, then it respects minority rights
but is not minority fair.\label{prop:NSWRuleProperties}
\end{prop}

\section{Single-round hiring}
\label{sec:SingleRoundHiring}

Until now, we evaluated rules from the perspective of whether they satisfy the desirable properties we introduced in the previous sections: aggregation independence, respecting minority rights, and minority fairness. While our analysis focuses on hirings potentially involving multiple rounds, one might wonder whether some of these problems would be present when the hiring is made in a single round. 

The property of aggregation independence does not imply anything regarding a single round of hiring. As we mentioned in section \ref{sec:SequentialAdjustedMinorReserves}, since the properties of respecting minority rights and minority fairness are only satisfied when they are satisfied for any number of rounds, the rules for which they are satisfied will also satisfy them for a single round of hiring. For the French policies, the results do not change.

\begin{rem}
\label{remark:singleRoundFrench}
Suppose, as in Proposition \ref{prop:FrenchAssignment}, that $M^*\subseteq W^*$, and that for every $w,w'\in M^*$, $s_w^O>s_{w'}^O\iff s_w^D>s_{w'}^D$. Policy 1 of the French assignment rule does not respect static minority rights, and Policy 2 respects static minority rights and is static minority fair.
\end{rem}

Next, consider the (sequential) use of minority reserves. When only a single round of hiring occurs, the rule satisfies all the desirable characteristics. Moreover, as we mentioned in section \ref{sec:SequentialAdjustedMinorReserves}, when there is a single round of hiring, it is equivalent to the sequential adjusted minority reserves.

\begin{rem}
\label{remark:singleRoundMinorityReserves}
The (sequential) use of minority reserves respects static minority rights and is static minority fair.
\end{rem}

The problems with minority reserves are not present under single hiring. As shown in Example \ref{example:sequentialMinorityReserves}, the issues reside on the fact that minority fairness requires that an asymmetric priority is given to minority candidates only when their proportion among those hired is below $m$. The sequential use of minority rights ``lacks memory'', in the sense that it always gives this asymmetric priority to minority workers, regardless of how much it is still needed given past hires. When only one round of hiring occurs, that is not a problem.

Regarding the Brazilian rule, the variable that determines its characteristics is $k$, that is, what is the number of workers from $W$ what will be put in the sets $TM$ and $O$. To see this, let $q_1$ be the number of hires in the single-round hiring, and let moreover $k=q_1$. One can easily verify that the hiring that will be made is the same as the one done by the (sequential) use of minority reserves. If $k>q_1$, on the other hand, we can have the situation shown in Example \ref{example:BrazilNotMinorityFair}. Therefore:

\begin{rem}
\label{remark:singleRoundBrazil}
Let $q_1$ be the number of hires that occur in a single round using the Brazilian rule. If $k=q_1$, then the rule respects static minority rights and is static minority fair. If $k>q_1$, then the rule respects static minority rights, but is not static minority fair.
\end{rem}

Finally, since the negative results for the NSW rule are based on a single round of hiring, we have the following remark:

\begin{rem}
\label{remark:singleRoundNSW}
When only one round of hiring occurs, the NSW rule is not fair. If one of the genders is deemed as a minority, then it respects static minority rights but is not static minority fair.
\end{rem}

\section{Multiple institutions\label{sec:Multiple-institutions}}

While often the hiring processes that we describe involve one or more positions in a single job specification -- and therefore the pool of applicants, or reserve list, being used only for that position\footnote{As we noted in section \ref{sec:Introduction}, most of the hiring in the public sector in France and Brazil, for example, requires the use of order of merit lists and reserve lists. While for some positions, such as police officers, it is natural to expect that the workers might be matched to different locations, many other positions are more specific and do not result in a pool of candidates shared by more than one job. Examples include the hiring of doctors with a specific specialty for a municipality with a single hospital, the role of economist in state companies, who only work at the headquarters, the entry-level diplomatic career, etc.} -- in many cases, a pool of workers is shared between multiple institutions
or locations. For example, in the hiring process for the Brazilian Federal police, workers may be allocated to different locations.\footnote{Source: Brazilian Department of Federal Police.}
In the selection process for the New Zealand police, the candidate's
preference is also taken into account when deciding which district
a worker who will be hired from the pool will go to:
\begin{quotation}
``The candidate pool is not a waiting list. The strongest candidates
are always chosen according to the needs and priorities of the districts.
The time it takes to get called up to college depends on your individual
strengths and the constabulary recruitment requirements in your preferred
districts. (...) We will look to place you into your preferred district
but you may also be given the option to work in another district where
recruits are needed most.''\footnote{Source: New Zealand Police (https://www.newcops.co.nz/recruitment-process/candidate-pool),
accessed in March 8th 2018.}
\end{quotation}
In this section, we evaluate how the fact that workers may be hired
by more than one institution affects the attainability of basic desirable
properties. Now, in addition to the set of workers $W$, there is also
a set of institutions $I=\left\{ i_{1},\ldots,i_{\ell}\right\}$, where $|I|\geq 3$.
Institutions make sequences of hires, and there is no simultaneity
in their hires: in each round only one institution may hire workers.
Therefore, when we describe a round, we now must determine not only
how many workers are hired, but also which institution those workers
will be assigned to. Some additional notation will be, therefore, necessary. A\textbf{
matching} $\mu$ is a function from $I\cup W$ to subsets of $I\cup W$
such that:
\begin{itemize}
\item $\mu\left(w\right)\in I\cup\left\{ \emptyset\right\} $ and $\left|\mu\left(w\right)\right|=1$
for every worker $i$,\footnote{We abuse notation and consider $\mu\left(w\right)$ to be an element
of $I$, instead of a set with an element of $I$.}
\item $\mu\left(i\right)\subseteq W$ for every institution $i$,
\item $\mu\left(w\right)=i$ if and only if $w\in\mu\left(i\right)$.
\end{itemize}
At the end of each round $r\geq1$, we define the\textbf{ matching}
of workers to institutions as a function $\mu^{r}$.

A \textbf{plural sequence of hires} $\Lambda$ is a list of pairs $\left(i,q\right)$, where
$i\in I$, and $q$ is the number of workers hired. A
plural sequence of hires $\Lambda=\langle\left(i_{1},3\right),\left(i_{3},2\right),\left(i_{1},2\right)\rangle$
, for example, represents the case in which in the first round institution
$i_{1}$ hires three workers, in the second round institution $i_{3}$
hires two workers, and then in the third round institution $i_{1}$
hires two workers. 

When considering hiring with multiple institutions, a rule, therefore, can be generalized to produce matchings instead of allocations. Given a pool of workers $W$, an initial matching $\mu^0$, and a plural sequence of hires $\Lambda=\langle (i_1,q_1),(i_2,q_2),\ldots,(i_k,q_k) \rangle$, a hiring rule $\Phi$ is derived from a set of institutional rules $\left(\Phi_i\right)_{i\in I}$ by returning the matching combining all institutional rules. That is, if $\Phi\left(W,\mu^0,\Lambda\right)=\mu$, then $\mu(i)=\Phi_i\left(W,\mu^0,\Lambda\right)$. Given $W$ and some matching $\mu^{t-1}$,  $\Phi_i\left(W,\mu^{t-1},\langle(i,q)\rangle\right)\subseteq W\backslash\bigcup_{i\in I}\mu^{t-1}(i)$. That is, it selects workers from $W$ who are not yet matched to some institution in $\mu^{t-1}$. Moreover:

\[\Phi_i\left(W,\mu^0,\Lambda\right)= \bigcup_{t=1}^{k} \Phi_i\left(W,\mu^{t-1},\langle(i_{t},q_{t})\rangle\right)\]

For any $t>0$, the matching $\mu^t$ is such that for all $i\neq i_t$,  $\mu^t(i)=\mu^{t-1}(i)$, but $\mu^t(i_t)=\mu^{t-1}(i_t)\cup \Phi_{i^t}\left(W,\mu^{t-1},\langle(i_{t},q_{t})\rangle\right)$. We assume that $\Phi_i\left(W,\mu,\langle(i',q)\rangle\right)=\emptyset$ whenever $i\neq i'$. That is, one institution cannot ``hire for another institution''. We also assume that if $\mu$ and $\mu'$ are such that $\cup_{i\in I}\mu(i)=\cup_{i\in I}\mu'(i)$ and $\mu(i^*)=\mu'(i^*)$, then $\Phi_{i^*}\left(W,\mu,\langle(i^*,q)\rangle\right)=\Phi_{i^*}\left(W,\mu',\langle(i^*,q)\rangle\right)$. That is, an institution $i^*$'s hiring decision can depend on the set of workers who were not yet hired and the workers who were already hired by $i^*$, but not on the identity of the workers who were hired by each other institution. 

Let $\mu^{\emptyset}$ be a matching in which for all $i\in I$, $\mu^{\emptyset}(i)=\emptyset$. We denote by $\Phi_i\left(W,\Lambda\right)$ and $\Phi\left(W,\Lambda\right)$ the values of $\Phi_i\left(W,\mu^{\emptyset},\Lambda\right)$ and $\Phi\left(W,\mu^{\emptyset},\Lambda\right)$. Finally, we abuse notation and if $\Lambda=\langle (i_1,q_1),(i_2,q_2),\ldots,(i_k,q_k)\rangle$, we can append the plural sequence of hires with the notation $\langle \Lambda,(i^*,q^*)\rangle\equiv\langle (i_1,q_1),(i_2,q_2),\ldots,(i_k,q_k),(i^*,q^*)\rangle$.

The example below clarifies these points.
\begin{example}
Consider a set of workers $W=\left\{ w_{1},w_{2},w_{3},w_{4},w_{5}\right\} $
with scores $s_{W}=\left(100,90,80,50,20\right)$, a set of institutions
$I=\left\{ i_{1},i_{2},i_{3}\right\} $ and let $\Phi$ be a rule
that, in any round, matches the highest scoring workers to the institution
in that round. Then if $\Lambda=\langle\left(i_{1},1\right),\left(i_{3},2\right),\left(i_{1},1\right)\rangle$, the matchings $\mu^{1}$, $\mu^{2}$ and $\mu^{3}$ produced at the
end of each round are:

\[
\mu^{1}=\begin{pmatrix}i_{1} & i_{2} & i_{3}\\
w_{1} & \emptyset & \emptyset
\end{pmatrix}\ \ \ \mu^{2}=\begin{pmatrix}i_{1} & i_{2} & i_{3}\\
w_{1} & \emptyset & \left\{ w_{2},w_{3}\right\} 
\end{pmatrix}\ \ \ \mu^{3}=\begin{pmatrix}i_{1} & i_{2} & i_{3}\\
\left\{ w_{1},w_{4}\right\}  & \emptyset & \left\{ w_{2},w_{3}\right\} 
\end{pmatrix}
\]
\end{example}
We will consider two properties for rules when there are multiple
institutions. The first is related to the desirability of workers.

\begin{defn}
A rule $\Phi$ satisfies\textbf{ common top} if there exists a
worker $w^{*}\in W$ such that, for every institution $i\in I$, $w^{*}\in\Phi_i\left(W,\langle(i,1)\rangle\right)$.
\end{defn}

In words, common top requires that there is at least one worker that,
whenever available, every institution would hire.

Next, we consider a weak notion of consistency across the hirings made by the institutions.
\begin{defn}
A rule $\Phi$ satisfies\textbf{ permutation independence} if for
any plural sequence of hires $\Lambda$ and any permutation of its elements
$\sigma\left(\Lambda\right)$, $\bigcup_{i\in I}\Phi_{i}\left(W,\Lambda\right)=\bigcup_{i\in I}\Phi_{i}\left(W,\sigma\left(\Lambda\right)\right)$.
\end{defn}

Permutation independence, therefore, simply requires that the set
of workers hired, regardless of where, should not change if we adjust
the order of hiring.

We also adapt the notion of aggregation independence to multiple institutions by requiring that each institution's rules are aggregation independent.

\begin{defn}
A rule $\Phi$ is\textbf{ aggregation independent} if for any $q\geq q_{1}\geq0$, sets of workers $W$, institution $i\in I$, and matching $\mu$, $\Phi_i\left(W,\mu,\langle (i,q)\rangle\right)=\Phi_i\left(W, \mu,\langle (i,q_1),(i,q-q_1)\rangle \right)$.
\end{defn}

The family of rules that we will use in our next result is elementary
but also very restrictive. 

A rule $\Phi$ is \textbf{single priority} if there exists a strict
ranking $\succ$ of the workers in $W$ such that when $\Lambda$
is any plural sequence of hires in which at least two different institutions make hires, for every $i\in I$:

\[\Phi_{i}\left(W,\langle\Lambda,\left(i,q\right)\rangle\right)=\Phi_{i}\left(W,\Lambda\right)\cup\max_{\succ}^{q}W\backslash\Phi_{i}\left(W,\Lambda\right)\]

Where, given a set $X$,  $\max_{\succ}^{q}X$ is the set with the top $q$ elements of
$X$ with respect to the ordering $\succ$. In words, a rule is single
priority if, whenever more than one institution make hires, all hirings from all institutions consist of hiring the
top workers, among the remaining ones, when all of these institutions
share a common ranking.

The result below shows that, for a wide range of applications, having
multiple institutions is incompatible
with most objectives a policymaker may have.
\begin{thm}
\label{thm:multipleInstitutionsSinglePriority}A rule satisfies common
top, aggregation independence and permutation independence if and only if it is a single priority
rule.
\end{thm}
Theorem \ref{thm:multipleInstitutionsSinglePriority} is a fundamentally
negative result. It shows that aggregation and permutation independence, both arguably simple desirable characteristics, are incompatible with institutions following different criteria when evaluating these candidates. This is true even when all institutions share the same scores for workers, but institutions might have different values of $m$ (the proportion of minorities that must be hired).\footnote{Take, for example, $I=\{i_1,i_2\}$, $W=\{w_1,w_2,w_3,w_4,w_5\}$, $M=\{w_1,w_2,w_5\}$, $s_{w_1}>s_{w_2}>s_{w_3}>s_{w_4}>s_{w_5}$, and the value of $m$ for the two institutions being $m_1=0.5$ and $m_2=0$. If both institutions use the sequential adjusted minority reserves and $i_1$ hires two workers before $i_2$ also hires two, worker $w_3$ is hired and $w_5$ is not. If the order is that $i_2$ hires first, then $w_5$ is hired and $w_3$ is not. A violation of permutation independence.}

Notice, however, that if the compositional objectives are interpreted as being applied to the entire set of workers hired by these institutions, as a whole, then we can simply use the sequential adjusted minority reserves, as defined in section \ref{sec:SequentialAdjustedMinorReserves}, every time an institution wants to hire a given number of workers. This procedure satisfies aggregation independence and is also permutation independent. Moreover, it satisfies a natural adaptation of what it means to respect minority rights and minority fairness. Instead of applying to whether workers are hired by a specific institution, it applies to being hired at \textit{some} institution.

\section{Conclusion\label{sec:Conclusion}}

In this paper, we evaluate a hiring method that is widely used around
the world, especially for public sector jobs, where institutions select
their workers over time from a pool of eligible workers. While the
simple and natural rule of sequential priority satisfies all desirable
characteristics, the addition of compositional objectives such as
affirmative action policies increases the complexity of the procedures.
We show that the rules being used in practical hiring processes,
as well as the direct application of minority reserves, fail fairness
or aggregation independence. When the compositional objectives can
be modeled as affirmative action for minorities, the sequential adjusted
minority reserves, which we introduced, is therefore the unique solution
that satisfies those desirable properties.

If multiple institutions hire from the same pool of applicants, however,
we show that the space for different hiring criteria between institutions, is highly restricted when a minimal requirement of independence is imposed.


\bibliographystyle{ecca}
\bibliography{CWAB}

\section*{Appendix}

\subsection*{Formal description of the rules}

For the descriptions in this section, consider as given a set $W$ of workers, a set $M\subseteq W$ of minority workers, a set $A$ of workers previously hired, a sequence of hires $q^r=\langle q_1,q_2,\ldots,q_k\rangle$, and a score profile $\mathbf{s_W}$. 

\subsubsection*{\textbf{Sequential Priority (SP rule)}}
\begin{quotation}

\textbf{\textit{Round 1:}} Let $W_{1}=W$. 
The highest scoring $q_{1}$ workers in $W_{1}$ are selected. 
Let $A_{1}$, be the set of selected workers, where for each $w\in A_{1}$ and each $w'\in W_{1}\setminus A_{1}$
we have $s_{w}>s_{w'}$, and $|A_{1}|=q_{1}$.

\textbf{\textit{Round $\bm{k>1}$:}} Let $W_{k}=W_{k-1}\setminus A_{k-1}$.
The highest scoring $q_{k}$ workers in $W_{k}$ are selected. Let $A_{k}$ be the
set of selected workers, where for each $w\in A_{k}$ and each $w'\in W_{k}\setminus A_{k}$
we have $s_{w}>s_{w'}$, and $|A_{k}|=q_{k}$.

The assignment selected by $SP$ rule is 

$\varphi^{SP}(W,\langle q_{1},\ldots,q_{r}\rangle)=\displaystyle\bigcup_{a\leq r}A_{a}$.
\end{quotation}

\subsubsection*{\textbf{Sequential Adjusted Minority Reserves (SA rule)}}
\begin{quotation}

\textbf{\textit{Round 1:}}

\textbf{Step 1.1:} Let $W_{1,1}=W$, $M_{1,1}=M\cap W_{1,1}$ and $q_{1,1}=\lceil m\times q_{1} \rceil$.
The highest scoring $\min\{q_{1,1},|M_{1,1}|\}$ workers  in $M_{1,1}$ are
selected. 
Let $A_{1,1}$ be the set of selected workers,
where $A_{1,1}\subseteq M_{1,1}$.

\textbf{Step 1.2:} Let $W_{1,2}=W_{1,1}\setminus A_{1,1}$, $M_{1,2}=M\cap W_{1,2}$
and $q_{1,2}=q_{1}-|A_{1,1}|$. 
The highest scoring $q_{1,2}$ workers in $W_{1,2}$ are selected.
Let $A_{1,2}$ be the set of selected workers.

\textbf{\textit{Round $\bm{k>1}$:}}

\textbf{Step k.1:} Let $W_{k,1}=W_{k-1,2}\setminus A_{k-1,2}$, $M_{k,1}=M\cap(W_{k-1,2}\setminus A_{k-1,2})$
and $q_{k,1}=\lceil\min\{\max\{m-\frac{\omega(A_{1,2})+\ldots+\omega(A_{k-1,2})}{q_{k}},0\}\times q_{k},|M_{k,1}|\}\rceil$.
The highest scoring $q_{k,1}$ workers in $M_{k,1}$ are selected.
Let $A_{k,1}$ be the set of selected workers.

\textbf{Step k.2:} Let $W_{k,2}=W_{k,1}\setminus A_{k,1}$, $M_{k,2}=M\cap W_{k,2}$
and $q_{k,2}=q_{k}-|A_{k,1}|$. 
The highest scoring $q_{k,2}$ workers are selected from $W_{k,2}$.
Let $A_{k,2}$ be the set of selected workers.

The assignment selected by the $SA$ rule is 

$\varphi^{SA}(W,A,\langle q_{1},\ldots,q_{r}\rangle)={\displaystyle \bigcup_{\substack{a\leq r\\
i\in\{1,2\}
}
}A_{a,i}}$.
\end{quotation}

\subsubsection*{\textbf{Sequential use of minority reserves (SM rule)}}
\begin{quotation}

\textbf{\textit{Round 1:}}

\textbf{Step 1.1:} Let $W_{1,1}=W$, $M_{1,1}=M\cap W_{1,1}$ and $q_{1,1}=\lceil m\times q_{1}\rceil$.
The highest scoring $\min\{q_{1,1},|M_{1,1}|\}$ workers are
selected from $M_{1,1}$. Let $A_{1,1}$ be the set of selected workers.

\textbf{Step 1.2:} Let $W_{1,2}=W\setminus A_{1,1}$, $M_{1,2}=M\cap W_{1,2}$
and $q_{1,2}=q_{1}-|A_{1,1}|$. 
The highest scoring $q_{1,2}$ workers are selected from $W_{1,2}$. Let $A_{1,2}$
be the set of selected workers.

\textbf{\textit{Round $\bm{k>1}$:}}

\textbf{Step k.1:} Let $W_{k,1}=W_{k-1,2}\setminus A_{k-1,2}$, $M_{k,1}=M\cap(W_{k-1,2}\setminus A_{k-1,2})$
and $q_{k,1}=\lceil m\times q_{k}\rceil$. 
The highest scoring $\min\{q_{k,1},|M_{k,1}|\}$ workers from $M_{k,1}$ are selected. Let $A_{k,1}$ be the set of selected
workers.

\textbf{Step k.2:} Let $W_{k,2}=W\setminus A_{k,1}$, $M_{k,2}=M\cap W_{k,2}$
and $q_{k,2}=q_{k}-|A_{k,1}|$. 
The highest scoring $q_{k,2}$ workers from $W_{k,2}$ are selected. Let $A_{k,2}$
be the set of selected workers.

The assignment produced by the $SM$ rule is $\varphi^{SM}(W,q^{r})={\displaystyle \bigcup_{\substack{a\leq r\\
i\in\{1,2\}
}
}A_{a,i}}$.
\end{quotation}

\subsubsection*{\textbf{Brazilian assignment rule (B rule)}}
\begin{quotation}
The rule first identifies a
large number $k$ (which is larger than the total number of vacancies
to be filled but no larger than $|W|$). Then two groups are identified: $(i)$ $TM$, which is the set with the top $k\times m$
minority workers: $TM\subseteq M$ with $|TM|=\lceil k\times m\rceil$
such that for each $w\in TM$ and each $w'\in M\setminus TM$, we have $s_{w}>s_{w'}$
and $(ii)$ $O$, which is the set with the top $k(1-m)$ workers among those who were not chosen
in $(i)$, that is: $O\subseteq W\setminus TM$ such that $|O|=\lfloor k(1-m)\rfloor$
and for each $w\in O$ and $w'\in W\setminus(O\cup TM)$, we have $s_{w}>s_{w'}$.
Within each round $a\leq r$, we have two steps.

\textbf{\textit{Round 1:}}

\textbf{Step 1.1:} Let $O_{1,1}=O$, $TM_{1,1}=TM$ and $q_{1,1}=\lceil m\times q_{1}\rceil$.
The highest scoring $\min\{q_{1,1},|TM_{1,1}|\}$ minority workers
are selected from $TM_{1,1}$. Let$A_{1,1}$ be the set of selected
workers.

\textbf{Step 1.2:} Let $O_{1,2}=O_{1,1}$, $TM_{1,2}=TM\setminus A_{1,1}$
and $q_{1,2}=q_{1}-|A_{1,1}|$. The highest scoring $q_{1,2}$ workers
are selected from $O_{1,2}$. Let $A_{1,2}$ be the set of selected
workers.

\textbf{\textit{Round $\bm{k>1}$:}}

\textbf{Step k.1:} Let $O_{k,1}=O_{k-1,2}\setminus A_{k-1,2}$, $TM_{k,1}=TM_{k-1,2}$
and $q_{k,1}=\lceil m\times q_{k}\rceil$. The highest scoring $\min\{q_{k,1},|TM_{k,1}|\}$
minority workers are selected from $TM_{k,1}$. Let $A_{k,1}$ be the
set of selected workers.

\textbf{Step k.2:} Let $O_{k,2}=O_{k,1}$, $TM_{k,2}=TM\setminus A_{k,1}$
and $q_{k,2}=q_{k}-|A_{k,1}|$. The highest scoring $q_{k,2}$ workers
are selected from $O_{k,2}$. Let$A_{k,2}$ be the set of selected
workers.

The assignment produced by the $B$ rule is $\varphi^{B}(W,q^{r})={\displaystyle \bigcup_{\substack{a\leq r\\
i\in\{1,2\}
}
}A_{a,i}}$.
\end{quotation}

\subsubsection*{\textbf{French assignment rule (F rule)}}
\begin{quotation}
Let $m$ be the target ratio of people with disabilities, $\mathbf{s_W^O}$ be a scoring profile for workers in the open competition and $\mathbf{s_W^D}$ be a scoring profile for workers in the competition for workers with disabilities.

\textbf{\textit{Round 1:}}

\textbf{Policy 1:} Let $W_{1,1}=W$, $M_{1,1}=M\cap W_{1,1}$. The highest scoring $\min\{q_{1,1},|W_{1,1}|\}$
workers, with respect to $s_{W}^{O}$, are selected from $W_{1,1}$.
Let $A_{1}$ be the set of selected workers.

\textbf{Policy 2:} Let $W_{1,1}=W$, $M_{1,1}=M$. The highest scoring $\min\{\lfloor(1-m)\times q_{1,1}\rfloor,|M_{1,1}|\}$
workers, with respect to $s_{W}^{D}$, are selected from $M_{1,1}$.
Let $A_{1,1}$ be the set of workers selected in this step. Then, the
highest scoring $\min\{\lceil m\times q_{1,1}\rceil,|W_{1,1}\setminus A_{1,1}|\}$
workers, with respect to $s_{W}^{O}$, are selected from $W_{1,1}\setminus A_{1,1}$. 
Let $A_{1,2}$ be the set of selected workers in this step, and let $A_{1}=A_{1,1}\cup A_{1,2}$.

\textbf{\textit{Round $\bm{k>1}$:}}

\textbf{Step k.1:} Let $W_{k,1}=W_{k-1,2}\setminus A_{k-1,2}$, $TA_{k,1}=\bigcup_{i=1}^{k-1}A_{i}$.
Let $q_{k,1}=min\left\{ max\left\{ m\times\left(\sum_{i=1}^{k}q_{i}\right)-\omega\left(TA_{k,1}\right),0\right\} ,|M_{k,1}|\right\} $.
The highest scoring $q_{k,1}$ workers, with respect to $s_{W}^{D}$,
are selected from $M_{k,1}$. Let $A_{k,1}$ be the set of workers selected
in this step.

\textbf{Step k.2:} Let $W_{k,2}=W_{k,1}\setminus A_{k,1}$, and $q_{k,2}=q_{k}-|A_{k,1}|$.
The highest scoring $q_{k,2}$ workers, with respect to $s_{W}^{O}$,
are selected from $W_{k,2}$. Let $A_{k,2}$ be the set of selected
workers, and $A_{k}=A_{k,1}\cup A_{k,2}$.

The assignment produced by the $F$ rule is $\varphi^{F}(W,q^{r})={\displaystyle \bigcup_{a\leq r}A_{a}}$.
\end{quotation}

\subsection*{Proofs}

\subsubsection*{Proof of Proposition \ref{prop:SeqMinorityReservesRespectsMinorityRightsNotConsistentNotFair}}

Example \ref{example:sequentialMinorityReserves} shows that the sequential use of minority reserves is neither aggregation independent nor fair. To see that it respects minority rights, notice that every time $q$ workers are hired, \textbf{at least} $m\times q$ minority workers are among them. As a result, a proportion of at least $m$ of the workers hired, at any point, is in $M$ and therefore the rule respects minority rights.  $\square$

\subsubsection*{Proof of Theorem \ref{thm:MSPMREquivMinorityRIghtsFairConsistent}}
First, we show that the SA rule respects minority rights and is minority fair. 

By definition, the SA rule \emph{respects minority rights}, at the
step $k.1$ of each round~$k$, selects minority workers to satisfy
the minimum requirement up to that round. Note that when there are
not enough minority workers, SA selects all the available minority workers.

Now, we show that the rule is minority fair.

Let $A\equiv \varphi^{SA} (W, \langle q_{1}, \ldots, q_{r}\rangle)$ be the selection made for the
problem. We want to show that $(i)$ for each $w,w'\in W\setminus M$,
if $w\in A$ and $w'\notin A$, then $s_{w}>s_{w'}$, $(ii)$ for each $w,w'\in M$,
if $w\in A$ and $w'\notin A$, then $s_{w}>s_{w'}$. $(iii)$ for each $w\in W\setminus M$
and $w'\in M$, if $s_{w}<s_{w'}$ and $w\in A$, then $w'\in A$, $(iv)$
if there is $w\in W\setminus M$ and $w'\in M$ with $s_{w}>s_{w'}$, $w\notin A$
and $w'\in A$, then $\omega(A)/|A|\leq m$.

First note that cases $(i)$, $(ii)$ and $(iii)$ hold trivially as
at step~$k.1$ of each round~$k$, the rule selects the highest scoring
workers in $M$, and in step~$k.2$ it selects the highest scoring workers.

Suppose, for contradiction, that there is $w\in W\setminus M$ and $w'\in M$
with $s_{w}>s_{w'}$, $w\notin A$ and $w'\in A$, but $\omega(A)/|A|>m$.
Note that $w'$ cannot be selected at step $k.2$ of any round~$k$,
as $w$ would have been selected as well. The only case in which the candidate $w'$
is selected is during step $\ell.1$ of some round~$\ell$. Since $s_{w}>s_{w'}$, $w\notin A$
and $w'\in A$, then $|top_{q}(W)\cap M|<m\times q$, where $q=\sum_{a\leq r}q_{a}$.\footnote{That is, the only way to hire a minority worker with a lower score
and not the non-minority with a higher score, is to satisfy the minority
requirements. As we mentioned earlier, the worker $w'$ is hired during
step $\ell.1$ of some round~$\ell$, where selection occurs among
minorities only.} Thus, at step $r.1$ of the last round~$r$, a selection is made
so that $|({\displaystyle \bigcup_{\substack{a<r\\
i\in\{1,2\}
}
}A_{a}^{i})\cup A_{k}^{1}|=m\times q}$. Since $|top_{q}(W)\cap M|<m\times q$, we have $A_{r}^{2}\cap M=\emptyset$.
Thus, we obtain $|A\cap M|=m\times q$ which contradicts our assumption.

For the next part, we need to introduce some concepts. First, we will say that given a set of workers $W$, minority workers $M\subseteq W$, score profile $\mathbf{s_W}$, $q\geq 0$ and $m\geq 0$, a set $W^*\subseteq W$ \textit{respects minority rights} if $W^*$ satisfies the same conditions that $\varphi\left(W,\langle q \rangle \right)$ must satisfy when $\varphi$ respects static minority rights (definition \ref{def:RespectsStaticMinorityRights}). Similarly, a set $W^*\subseteq W$ is \textit{minority fair} if $W^*$ satisfies the same conditions that $\varphi\left(W,\langle q \rangle \right)$ must satisfy when $\varphi$ is static minority fair (definition \ref{def:StaticMinorityFair}). 

To show that if a rule is minority fair and respects minority rights it is the $SA$ rule, we show that, for any given number of workers to be hired $q$, there is only one set of workers of that size that is minority fair and respects minority rights.

\begin{lemma}
\label{lem:UniqueMinoritySet}
For any given set of workers $W$, minority workers $M\subseteq W$, score profile $\mathbf{s_W}$, $q\geq 0$ and $m\geq 0$, there exists only one set $W^*\subseteq W$ that respects minority rights, is minority fair and such that $\left|W^*\right|=q$.
\end{lemma}

\begin{proof}
First, note that property (i) of minority fairness implies that if a set $W^*$ is minority fair, it contains the set $top_{\omega(W^*)}(M)$, that is, the top $\omega(W^*)$ highest scoring workers in $M$, and the set $top_{q-\omega(W^*)}(W\backslash M)$, the $q-\omega(W^*)$ highest scoring workers in $W\backslash M$, both with respect to $\mathbf{s_W}$.

Suppose, for contradiction, that there are $W^1\subseteq W$ and $W^2\subseteq W$, where both $W^1$ and $W^2$ respect minority rights and are minority fair, $\left|W^1\right|=\left|W^2\right|=q$, and $W^1\neq W^2$. 

Note first that if $|M|<m\times q$, respecting minority rights implies that $M\subset W^1$ and $M\subset W^2$. Minority fairness implies, moreover, that $top_{q-|M|}(W\backslash M)\subseteq W^1$ and $top_{q-|M|}(W\backslash M)\subseteq W^2$. But then $W^1=W^2$, a contradiction. It must be, therefore, that $|M|\geq m\times q$.

Next, note that minority fairness implies that $\omega(W^1)\neq \omega(W^2)$. To see that, notice that if $\omega(W^1)=\omega(W^2)=m^*$, $top_{m^*}(M)\subseteq W^1$, $top_{q-m^*}(W\backslash M)\subseteq W^1$, $top_{m^*}(M)\subseteq W^2$, and $top_{q-m^*}(W\backslash M)\subseteq W^2$. But this would imply that $W^1=W^2$, a contradiction.

Suppose now, without loss of generality, that $\omega(W^1)>\omega(W^2)$. Since $W^2$ respects minority rights, $\omega(W^2)\geq m\times q$, and therefore $\omega(W^1)>\omega(W^2)\geq m\times q$. Therefore, there is a worker $w^*_1\in W\backslash M$ such that $w^*_1\in W^2$ and $w^*_1\not\in W^1$, and a worker $w^*_2\in M$ such that $w^*_2\in W^1$ and $w^*_2\not\in W^2$.

We have two cases to consider. First, suppose that $s_{w^*_1}>s_{w^*_2}$. This would violate $W^1$ being minority fair, since $m(W^1)> m\times q$ and $w^*_1\not\in W^1$. Then it must be that $s_{w^*_2}>s_{w^*_1}$. But then, since $W^2$ is minority fair, condition (ii) implies that $w^*_2\in W^2$, a contradiction.

We conclude, therefore, that $W^1\neq W^2$ is false, proving uniqueness.
\end{proof}

Since the SA rule respects minority rights and is minority fair, lemma \ref{lem:UniqueMinoritySet} implies that this is the only such rule.

$\square$

\subsubsection*{Proof of Theorem \ref{thm:SAMAggIndep}}

Let $\varphi$ be a rule that is static minority fair, satisfies static minority rights, and is aggregation independent. Let $\lambda^*$ be any sequence of hires.

We will follow by induction on the rounds in $\lambda^*$. First, the base $\langle q_1\rangle$: from lemma \ref{lem:UniqueMinoritySet}, there is a unique set $W^1\subseteq W$ that is minority fair and respects minority rights. Both $\varphi$ and $\varphi^{SA}$ are static minority fair and respect static minority rights. Therefore, $\varphi(W,\langle q_1 \rangle)=\varphi^{SA}(W,\langle q_1 \rangle)=W^1$.

For the induction step, assume that $\varphi(W,\langle q_1,q_2,\ldots,q_{\ell} \rangle)=\varphi^{SA}(W,\langle q_1,q_2,\ldots,q_{\ell} \rangle)$. Since $\varphi$ is aggregation independent, the following is true:

\[\varphi\left(W,\langle q_1,q_2,\ldots,q_{\ell} \rangle\right) = \varphi\left(W,\langle q \rangle\right)\]

where $q=\sum_{i=1}^{\ell}q_i$. Let $H=\varphi\left(W,\langle q \rangle\right)$. Aggregation independence of $\varphi$ implies, moreover, that:

\[\varphi\left(W,\langle q \rangle\right) \cup \varphi\left(W,H,\langle q_{\ell+1} \rangle\right) = \varphi\left(W,\langle q+q_{\ell+1} \rangle\right)\text{    }(*)\]

Since both $\varphi$ and $\varphi^{SA}$ are static minority fair and respect static minority rights, our claim above implies that $\varphi\left(W,\langle q \rangle\right)=\varphi^{SA}\left(W,\langle q \rangle\right)$ and $\varphi\left(W,\langle q+q_{\ell+1} \rangle\right)=\varphi^{SA}\left(W,\langle q+q_{\ell+1} \rangle\right)$. Since workers cannot be hired more than once, $\varphi\left(W,\langle q \rangle\right) \cap \varphi\left(W,H,\langle q_{\ell+1} \rangle\right)=\emptyset$. 

Therefore, there is a unique value of $\varphi\left(W,H,\langle q_{\ell+1} \rangle\right)$ that satisfies the equality $(*)$ above, implying that $\varphi\left(W,H,\langle q_{\ell+1} \rangle\right)=\varphi^{SA}\left(W,H,\langle q_{\ell+1} \rangle\right)$, and therefore that:

\[\varphi\left(W,\langle q_1,q_2,\ldots,q_{\ell},q_{\ell+1} \rangle\right)=\varphi^{SA}\left(W,\langle q_1,q_2,\ldots,q_{\ell},q_{\ell+1} \rangle\right)\]

finishing our proof. $\square$

\subsubsection*{Proof of Proposition \ref{prop:FrenchAssignment}}

Let $W^{*}=\{w_{1},w_{2},w_{3},w_{4},w_{5}\}$
with scores $s_{W}=(50,40,30,20,10)$. For simplicity, we will use $m=0.5$.

Consider first the case $M^{*}=\left\{ w_{3},w_{4}\right\} $. If $q=2$, $\varphi^{F_{1}}(\left\{ W^{*},M^{*}\right\} ,q)=\left\{ w_{1},w_{2}\right\} $,
which fails to satisfy minority rights.

Consider now the case $M^{*}=\left\{ w_{4},w_{5}\right\} $. Consider
two possibilities: $q_{1}=q_{2}=2$ and $q=4$. Then $\varphi^{F_{1}}\left(\left\{ W^{*},M^{*}\right\} ,\langle q_{1},q_{2}\rangle \right)=\left\{ w_{1},w_{2},w_{4},w_{5}\right\} $
but $\varphi^{F_{1}}\left(\left\{ W^{*},M^{*}\right\} ,q\right)=\left\{ w_{1},w_{2},w_{3},w_{4}\right\} $,
a violation of \emph{aggregation independence.}

It is easy to see that the rule that results from policy 2, under
the given assumptions, is equivalent to the sequential adjusted minority reserves rule. Therefore,
Policy 2 of the French assignment rule respects minority rights, is
aggregation independent, and is minority fair. $\square$

\subsubsection*{Proof of Proposition \ref{prop:BrazilianRuleNotMinorityFair}} We will show that the Brazilian rule respects minority rights and is aggregation independent, but fails to be minority fair. 

By assumption, no more than $k$ workers may be hired in total. Therefore, for any $q$ workers to be hired in any given round there should be at least $\lceil q\times m \rceil$ minority workers in $TM$ and $q-\lceil q\times m \rceil$ workers in $O$. As a result, the Brazilian rule acts as two parallel sequential priority rules: one in $TM$ and one in $O$. Therefore, the combination of both is evidently aggregation independent. Next, notice that again because of the assumption on the value of $k$, $\left|M\right|\geq m\times\sum_{i=1}^{t}q_{i}$. Moreover, since for any $q\in \left\{ q_1,\ldots,q_t \right\}$ there are  at least $\lceil q\times m \rceil$ minority workers in $TM$, $\omega\left(\varphi\left(W,\langle q_{1},\ldots,q_{t}\rangle \right)\right)\geq m\times \sum q_i$ and by assumption on $k$, $\left|\varphi\left(W,\langle q_{1},\ldots,q_{t}\rangle \right)\right|=\sum q_i$ therefore $\omega\left(\varphi\left(W,\langle q_{1},\ldots,q_{t}\rangle \right)\right)/\left|\varphi\left(W,\langle q_{1},\ldots,q_{t}\rangle \right)\right|\geq m$, implying that the Brazilian rule respects minority rights.
Finally, example \ref{example:BrazilNotMinorityFair} shows that the rule is not minority fair.  $\square$

\subsubsection*{Proof of Proposition \ref{prop:NSWRuleProperties}}

Example \ref{example:NSWNotfairNotMinorityFair} shows that the NSW rule is neither fair nor minority fair. Moreover, since in our results we assume that the number of men and women are always large enough, the NSW consists of two parallel sequential priority hirings (one for males, the other for female workers), and therefore satisfies aggregation independence. Finally, it respects minority rights, since the number of male and female workers hired is always the same. $\square$

\subsubsection*{Proof of Theorem \ref{thm:multipleInstitutionsSinglePriority}}
The single priority rule satisfying common top, aggregation independence and permutation independence is straightforward to see. 

Denote by \textit{sequence of single hirings} a plural sequence of hires of the form $\Lambda=\langle(i^1,1),(i^2,1),(i^3,1),\ldots\rangle$. That is, every hire made by any institution in any round consists of only one worker.

\begin{claim}
Let $W$ be a set of workers, and $\Phi$ be a rule that satisfies common top and permutation independence. There exists a ranking $\succ^*$ over $W$ such that for any sequence of single hirings $\Lambda$, $\Phi(W,\Lambda)=\Phi^{\succ^*}(W,\Lambda)$, where $\Phi^{\succ^*}$ is the single priority rule that uses $\succ^*$.
\end{claim}
\begin{proof}
We will prove by induction on the number of hires in a plural sequence of hires. That is, we will show that there exists a ranking $\succ^*$, independent of the sequence of hires, that is followed by $\Phi$ as a single priority.

In the remaining steps of the proof, the set $W$ and the rule $\Phi$ are given, and so for any plural sequence of hires $\Lambda$, we will use the notation $\left\{\Lambda\right\}$ to represent the set $\bigcup_{i\in I}\Phi_i\left(W,\Lambda\right)$. That is, $\left\{\Lambda\right\}$ is the set of workers in $W$ hired by some institution under $\Phi$ after the sequence of hires $\Lambda$. Since we will only look at single hirings, we will represent plural sequences of hires as sequences of institutions and use $\langle i_1,i_2,\ldots\rangle$ to represent $\langle (i_1,1),(i_2,1),\ldots\rangle$. 

We will use \textbf{(PI)} to indicate that we used the property of \textit{permutation independence} of $\Phi$,  \textbf{(AI)} to indicate that we used \textit{aggregation independence}, and \textbf{(CT)} for \textit{common top}.

Moreover, we will use \textbf{(P*)} to indicate that we are using the following fact:

\begin{quotation}
If $\Lambda_1$, $\Lambda_2$, and $i\in I$ are such that $\left\{\Lambda_1\right\}=\left\{\Lambda_2\right\}$ and $\Phi_i\left(W,\Lambda_1\right)=\Phi_i\left(W,\Lambda_2\right)$, then $\Phi_i\left(W,\langle\Lambda_1,i\rangle\right)=\Phi_i\left(W,\langle\Lambda_2,i\rangle\right)$. That is, if $\Lambda_1$ and $\Lambda_2$ are such that institution $i$ hires the same set of workers, and the set of workers remaining after all of the hires in both plural sequences of hires is the same, then $i$ would hire the same worker after both $\Lambda_1$ and $\Lambda_2$. This comes directly from the definition of a hiring rule $\Phi_i$.
\end{quotation}

\textbf{Induction Base}
The induction base is the case where the smallest number of hires is made while still having at least two institutions hiring. Therefore $|\Lambda|=2$. Suppose that the claim is not true. That is, there might be plural sequences of hires with two hires that cannot be explained by a ranking $\succ^*$ over $W$. That implies that there are $\Lambda_1\neq\Lambda_2$, where $\Lambda_1=\langle i_1,i_2\rangle$, $\Lambda_2=\langle i_3,i_4\rangle$, and $\left\{\Lambda_1\right\}\neq\left\{\Lambda_2\right\}$. 

Since the sequences of hires involve at least two institutions, $i_1\neq i_2$ and $i_3\neq i_4$. Since $\Lambda_1\neq\Lambda_2$, there are two cases to consider: (i) $i_1\neq i_3$, and (ii) $i_2\neq i_4$. Consider (i). By \textbf{(PI)}, $\left\{\langle i_1,i_2\rangle\right\}=\left\{\langle i_2,i_1\rangle\right\}$. By \textbf{(P*)}, \textbf{(CT)} and the fact that $i_1\neq i_3$, $\left\{\langle i_2,i_1\rangle\right\}=\left\{\langle i_2,i_3\rangle\right\}$. By \textbf{(PI)}, $\left\{\langle i_2,i_3\rangle\right\}=\left\{\langle i_3,i_2\rangle\right\}$. By \textbf{(P*)}, \textbf{(CT)} and the fact that $i_3\neq i_4$, $\left\{\langle i_3,i_2\rangle\right\}=\left\{\langle i_3,i_4\rangle\right\}$. But then $\left\{\langle i_1,i_2\rangle\right\}=\left\{\langle i_3,i_4\rangle\right\}$, a contradiction. For case (ii), \textbf{(PI)} implies that $\left\{\langle i_1,i_2\rangle\right\}=\left\{\langle i_2,i_1\rangle\right\}$ and $\left\{\langle i_3,i_4\rangle\right\}=\left\{\langle i_4,i_3\rangle\right\}$, which makes this case equivalent to (i).

\textbf{Induction Step}

We now assume that for every sequence of single hirings $\Lambda$ such that $|\Lambda|\leq k$, the rule $\Phi$ hires according to the ranking $\succ^*$. We will use \textbf{(IA)} to indicate that we are using this \textit{induction assumption}. 

Suppose now that the claim is not true. That is, there are sequences of single hirings $\Lambda_1$, $\Lambda_2$, such that $|\Lambda_1|=|\Lambda_2|=k$, and institutions $i_1,i_2\in I$, for which $\left\{\langle \Lambda_1,i_1\rangle\right\}\neq\left\{\langle \Lambda_2,i_2\rangle\right\}$. There are two cases to consider.

\textbf{Case (i): $i_1\neq i_2$.}

Let $\Lambda_1^a$ be a sequence of single hires, such that:
    \begin{itemize}
        \item $|\Lambda_1^a|=k$
        \item The rounds in which $i_1$ hires in $\Lambda_1^a$ are exactly the same in which $i_1$ hires in $\Lambda_1$, if any.
        \item For the rounds in which $i_1$ does not hire:
        \begin{itemize}
            \item Let $i_3$ be the institution hiring at the first round in which $i_1$ doesn't hire  (note that this must exist, since $\Lambda_1$ has hirings from at least two institutions).
            \item Let $i_2$ be the institution hiring at every other round of $\Lambda_1^a$, if any.
        \end{itemize}
    \end{itemize}

In $\Lambda_1^a$, therefore, hires made by $i_1$, if any, are the same as in $\Lambda_1$, there is exactly one hire by $i_2$, and the remaining hires, if any, are made by $i_3$.

By \textbf{(IA)} and \textbf{(P*)}, $\left\{\langle \Lambda_1,i_1\rangle\right\}=\left\{\langle \Lambda_1^a,i_1\rangle\right\}$. Next, let $\Lambda_1^b$ be exactly as $\Lambda_1^a$, except that the single place where $i_3$ is is replaced by $i_1$. By \textbf{(PI)}, $\left\{\langle \Lambda_1^a,i_1\rangle\right\}=\left\{\langle \Lambda_1^b,i_3\rangle\right\}$. Notice that $\Lambda_1^b$ contains all the hires made by $i_1$ in $\Lambda_1$, in addition to one extra hire from $i_1$. All other hires in $\Lambda_1^b$ are made by $i_2$. That is, there is no hire from $i_3$ in $\Lambda_1^b$.

Next, let $\Lambda_1^c$ a sequence of hires where the hires in $\Lambda_1^c$ are exactly the same as $\Lambda_2$, except that:

\begin{itemize}
    \item Every round in which $i_3$ hires in $\Lambda_2$, the hire is made by $i_1$ instead,
    \item Denote by $t^*$ the first round in which $i_3$ does not hire in $\Lambda_2$. Note that this must exist, since $\Lambda_2$ has hirings from at least two institutions. Let $i_2$ hire in round $t^*$ in $\Lambda_1^c$ instead.
\end{itemize}

 Notice, therefore, that there is no hire from $i_3$ in $\Lambda_1^c$. By \textbf{(IA)} and \textbf{(P*)}, therefore, $\left\{\langle \Lambda_1^b,i_3\rangle\right\}=\left\{\langle \Lambda_1^c,i_3\rangle\right\}$.

Next, let $\Lambda_1^d$ be exactly as $\Lambda_1^c$, except that the $i_2$ in round $t^*$ is replaced by $i_3$. By \textbf{(PI)}, $\left\{\langle \Lambda_1^c,i_3\rangle\right\}=\left\{\langle \Lambda_1^d,i_2\rangle\right\}$. Notice that the rounds in which $i_2$ hires in $\Lambda_1^d$ are exactly the same as in $\Lambda_2$, and as a result, \textbf{(P*)} and \textbf{(IA)} imply that the last hire made by $i_2$ in $\langle \Lambda_1^d,i_2\rangle$ is the same as in $\langle \Lambda_2,i_2\rangle$. Not only that, \textbf{(IA)} implies that the set of workers hired in the first $k$ hires are the same, and therefore  $\left\{\langle \Lambda_1^d,i_2\rangle\right\}=\left\{\langle \Lambda_2,i_2\rangle\right\}$, implying that $\left\{\langle \Lambda_1,i_1\rangle\right\}=\left\{\langle \Lambda_2,i_2\rangle\right\}$, a contradiction.

\textbf{Case (ii): $i_1=i_2$.}

We will use three institutions in the following steps: $i,i_a,i_b$, where $i=i_1=i_2$, and $i\neq i_a\neq i_b$. Let $\Lambda_1^a$ be a sequence of single hires, such that $|\Lambda_1^a|=k$, the rounds in which $i$ hires in $\Lambda_1$ are exactly the same in which $i$ hires in $\Lambda_1^a$, if any. Moreover, let $i_b$ be the institution hiring at the first round in which $i$ doesn't hire in $\Lambda_1$ (note that this must exist, since $\Lambda_1$ has hirings from at least two institutions), and $i_a$ be the institution hiring in every other rounds, if any. By \textbf{(IA)} and \textbf{(P*)}, $\left\{\langle \Lambda_1,i\rangle\right\}=\left\{\langle \Lambda_1^a,i\rangle\right\}$.

Next, let $\Lambda_1^b$ be exactly as $\Lambda_1^a$, replacing the single place where $i_b$ is by $i$. By \textbf{(PI)}, $\left\{\langle \Lambda_1^a,i\rangle\right\}=\left\{\langle \Lambda_1^b,i_b\rangle\right\}$. 

Next, let $\Lambda_1^c$ be exactly as $\Lambda_2$, except that every hire made by $i_b$, if any, is made instead by $i_2$. Moreover, let $i$ be the institution hiring at the first round in which $i_b$ doesn't hire in $\Lambda_2$ (note that this must exist, since $\Lambda_2$ has hirings from at least two institutions). Denote this round by $t^*$. Notice, therefore, that there is no hire from $i_b$ in $\Lambda_1^c$. By \textbf{(IA)} and \textbf{(P*)}, therefore, $\left\{\langle \Lambda_1^b,i_b\rangle\right\}=\left\{\langle \Lambda_1^c,i_b\rangle\right\}$.

Next, let $\Lambda_1^d$ be exactly as $\Lambda_1^c$, except that the $i$ in round $t^*$ is replaced by $i_b$. By \textbf{(PI)}, $\left\{\langle \Lambda_1^c,i_b\rangle\right\}=\left\{\langle \Lambda_1^d,i\rangle\right\}$.

Notice that the rounds in which $i$ hires in $\Lambda_1^d$ are exactly the same as in $\Lambda_2$, and as a result, \textbf{(P*)} and \textbf{(IA)} imply that the last hire made by $i$ in $\langle \Lambda_1^d,i\rangle$ is the same as in $\langle \Lambda_2,i\rangle$. Not only that, \textbf{(IA)} implies that the set of workers hired in the first $k$ hires are the same, and therefore  $\left\{\langle \Lambda_1^d,i\rangle\right\}=\left\{\langle \Lambda_2,i\rangle\right\}$, implying that $\left\{\langle \Lambda_1,i\rangle\right\}=\left\{\langle \Lambda_2,i\rangle\right\}$, a contradiction.
\end{proof}

Finally, let $\Lambda=\langle (i_1,q_1),(i_2,q_2),\ldots, (i_k,q_k)\rangle$ be any plural sequence of hires and $\Phi^*$ be a rule that satisfies common top, permutation independence, and aggregation independence. By \textbf{(AI)}:

\[\Phi^*\left(W,\Lambda\right)=\Phi^*\left(W,\langle\underbrace{\left(i_1,1\right),\ldots,\left(i_1,1\right)}_{q_1\text{ times}},\ldots, \underbrace{\left(i_k,1\right),\ldots,\left(i_k,1\right)}_{q_k\text{ times}}\rangle\right)\]

That is, aggregation independence implies that each hire from an institution can be split into single hires without changing the workers that are chosen, round by round.\footnote{Notice that the property of aggregation independence holds for any initial matching $\mu$.} Our claim above implies, therefore, that the rule $\Phi^*$ must be single priority, finishing our proof.

$\square$

\end{document}